\documentclass[pra,aps,groupedaddress,twocolumn,floatfix,nofootinbib]{revtex4}
\usepackage{amsmath,amsthm,amsfonts,amssymb,graphics,graphicx,epsfig,color,times,natbib,subfigure}
\usepackage{bbm} 

\usepackage[strict]{changepage}

\newcommand{\braket}[2]{\mbox{$ \langle #1 | #2 \rangle $}}
\newcommand{\sandwich}[3]{\mbox{$ \langle #1 | #2 | #3 \rangle $}}
\newcommand{\ket}[1]{\mbox{$ | #1 \rangle $}}
\newcommand{\bra}[1]{\mbox{$ \langle #1 | $}}

\newcommand{\be}{\begin{equation}}
\newcommand{\ee}{\end{equation}}
\newcommand{\ba}{\begin{eqnarray}}
\newcommand{\ea}{\end{eqnarray}}

\newcommand{\one}{\leavevmode\hbox{\small1\normalsize\kern-.33em1}}

\newcommand{\moy}[1]{\langle #1 \rangle}

\newtheorem{thm}{Theorem}

\newtheorem{lem}[thm]{Lemma}

\begin{document}

\title{How well can one jointly measure two incompatible observables on a given quantum state?}

\author{Cyril Branciard}
\affiliation{Centre for Engineered Quantum Systems and School of Mathematics and Physics, The University of Queensland, St Lucia, QLD 4072, Australia}

\date{\today}

\begin{abstract}
Heisenberg's uncertainty principle is one of the main tenets of quantum theory. Nevertheless, and despite its fundamental importance for our understanding of quantum foundations, there has been some confusion in its interpretation: although Heisenberg's first argument was that the measurement of one observable on a quantum state necessarily disturbs another incompatible observable, standard uncertainty relations typically bound the indeterminacy of the outcomes when either one or the other observable is measured. In this paper, we quantify precisely Heisenberg's intuition. Even if two incompatible observables cannot be measured together, one can still approximate their joint measurement, at the price of introducing some errors with respect to the ideal measurement of each of them. We present a new, tight relation characterizing the optimal trade-off between the error on one observable versus the error on the other. As a particular case, our approach allows us to characterize the disturbance of an observable induced by the approximate measurement of another one; we also derive a stronger error-disturbance relation for this scenario.
\end{abstract}

\maketitle

The discovery and development of quantum theory have generated passionate debates amongst its founding fathers. The surprising features of the theory---e.g., its probabilistic nature, its uncertainty principle~\cite{heisenberg27} or its nonlocality~\cite{EPR,bell_book}---were indeed too counter-intuitive to satisfy all physicists: Einstein, for instance, famously argued that ``God does not play dice''~\cite{God_dice}, and could not accept the apparent ``spooky action at a distance''~\cite{spooky_action} that seemed to be allowed by the theory.
Interestingly, it has since then been realized that what first seemed to be limitations of the theory---the impossibility to perfectly predict measurement outcomes and to explain them with local hidden variables---can turn out to allow for useful applications for information processing, such as quantum cryptography for instance~\cite{QKD_review}. With the advent of quantum information science, it becomes all the more essential to clarify what can or cannot be done quantum mechanically.

The well-known uncertainty principle is typically expressed in terms of ``uncertainty relations''. To fix the notations, let us define the standard deviations $\Delta A, \Delta B$ of two observables $A$ and $B$ in the state $\ket{\psi}$ as
\ba
\Delta A &=& \sandwich{\psi}{\,(A - \moy{A})^2\,}{\psi}^{1/2}, \label{def_sigA} \\
\Delta B &=& \sandwich{\psi}{\,(B - \moy{B})^2\,}{\psi}^{1/2}, \label{def_sigB}
\ea
with $\moy{A} = \sandwich{\psi}{A}{\psi}$ and $\moy{B} = \sandwich{\psi}{B}{\psi}$, and the ``value'' of the commutator $[A,B] = AB - BA$ in the state $\ket{\psi}$, divided by $2i$, as
\ba
C_{\!AB} &=& \frac{1}{2i} \sandwich{\psi}{[A,B]}{\psi} . \label{def_CAB}
\ea
Robertson's well-known uncertainty relation~\cite{robertson_UR} then imposes that
\ba
\Delta A \ \Delta B &\geq& |C_{\!AB}|. \label{robertson_UR}
\ea
Such uncertainty relations are often wrongly interpreted---even, historically, by some of the most illustrious authors~\cite{heisenberg1949_book,vonNeumann_book,bohm_book,bohr}---as saying that one cannot jointly measure the observables $A$ and $B$ on the state $\ket{\psi}$ when $C_{\!AB} \neq 0$, or that the measurement of one observable necessarily disturbs the other. Although this last observation corresponds indeed to Heisenberg's intuition~\cite{heisenberg27}, this is actually not what standard uncertainty relations imply, let-alone quantify~\cite{ballentine}. Rather than referring to joint (or successive) measurements of two observables on one state, they indeed bound the statistical deviations of the measurement results of $A$ and $B$, when each measurement is performed many times on several independent, identically prepared quantum states.

In this paper we aim instead at precisely quantifying Heisenberg's original formulation of the uncertainty principle. Even if two observables $A$ and $B$ are incompatible and can indeed not be jointly measured on a state $\ket{\psi}$, it is still possible to approximate their joint measurement. How good can such an approximation be? What is the optimal trade-off between the error induced on the measurement of $A$ and the error on $B$? Between the error in the approximation of one observable and the disturbance implied on the other?
We answer these questions below by deriving new, tight error-trade-off and error-disturbance relations.

\section*{APPROXIMATE JOINT MEASUREMENTS}

Let us start by setting up our general framework for approximate joint measurements. Our presentation is inspired by those of Refs.~\cite{ozawa03,hall04,ozawa04}, and is restricted here to the basics; more details are given in the Supporting Information (Part~A).

In order to approximate the measurement of an observable $A$ on a quantum system in the state $\ket{\psi}$ (in some Hilbert space ${\cal H}$), a general strategy consists in measuring another, ``approximate'' observable ${\cal A}$, possibly on an extended Hilbert space---i.e., on the joint system composed of the state $\ket{\psi} \in {\cal H}$, and of an ancillary system in the state $\ket{\xi}$ of another Hilbert space ${\cal K}$. In this picture, the impossible joint measurement of two incompatible observables $A$ and $B$ on $\ket{\psi}$ can thus be approximated by the perfect joint measurement of two compatible (i.e., commuting) observables ${\cal A}$ and ${\cal B}$ on $\ket{\psi,\xi} = \ket{\psi} \otimes \ket{\xi} \in {\cal H} \otimes {\cal K}$. Note that in full generality, we do not assume {\emph{a priori}} (for now at least) that ${\cal A}$ and ${\cal B}$ must have the same spectrums as $A$ and $B$.

Following Ozawa~\cite{Ozawa_1991,Ozawa03_physical_content,ozawa03,ozawa04,Ozawa04_gen_measurements}, we characterize the quality of the approximations ${\cal A}$ and ${\cal B}$ of $A$ and $B$, respectively, by defining the {\it root-mean-square (rms) errors}
\ba
\epsilon_{\cal A} &=& \sandwich{\psi,\xi}{\,({\cal A} - A \otimes \one)^2\,}{\psi,\xi}^{1/2} , \label{def_epsA} \\
\epsilon_{\cal B} &=& \sandwich{\psi,\xi}{\,({\cal B} - B \otimes \one)^2\,}{\psi,\xi}^{1/2} . \label{def_epsB}
\ea
These rms errors, which generalize standard definitions in classical estimation theory~\cite{van2001detection}, quantify the statistical deviations between the approximations ${\cal A}$ and ${\cal B}$, and the ideal measurements of $A$ and $B$. We refer to Refs.~\cite{hall03_POM,Ozawa03_physical_content,Ozawa04_gen_measurements,lund_wiseman} for discussions on the motivations and appropriateness of such definitions. There has been a controversy~\cite{werner04_QIC,Koshino2005191} on the question whether these quantities were experimentally accessible; two different indirect methods have nevertheless been proposed~\cite{Ozawa04_gen_measurements,lund_wiseman}, and recently implemented~\cite{erhart12,rozema12_steinberg_expment}.

\section*{$\quad \quad \ \ \ $ ERROR-TRADE-OFF RELATIONS \newline FOR JOINT MEASUREMENTS}

The fact that quantum theory forbids perfect joint measurements of incompatible observables implies that the rms errors $(\epsilon_{\cal A},\epsilon_{\cal B})$ can in general not take arbitrary values. Some limitations on their possible values have been obtained previously~\cite{ArthursKelly65,ArthursGoodman88,Ishikawa_1991,Ozawa_1991,ozawa03,ozawa04}, which we review below. For historical reasons, such limitations are often referred to as ``uncertainty relations'' (for joint measurements).
We will keep this terminology when we refer to previously derived relations; however, since such relations are not strictly speaking about uncertainty but about errors in the approximation of joint measurements, we will prefer the terminology ``error-trade-off relations (for joint measurements)''.

\subsection*{The Heisenberg-Arthurs-Kelly relation}

In his seminal paper~\cite{heisenberg27}, Heisenberg argued that the measurement of the position $q$ of a particle necessary implies a disturbance $\eta_p$ on its momentum $p$, and that this disturbance is all the more important as the precision of the measurement of $q$ is large (or as the ``error'' $\epsilon_q$ is small)---so that $\epsilon_q \, \eta_p \sim h$, where $h$ is the Planck constant.

The formalization of Heisenberg's intuition rapidly lead to the derivation of general uncertainty relations in terms of standard deviations [as in~\eqref{robertson_UR}] rather than of error and disturbance. Nevertheless, it is commonly believed that a similar relation to Robertson's should also restrict the possible values of the errors $\epsilon_{\cal A}$ and $\epsilon_{\cal B}$ on $A$ and $B$ in an approximate joint measurement, in such a way that
\ba
\epsilon_{\cal A} \ \epsilon_{\cal B} & \geq & |C_{\!AB}| .
\label{heisenberg_relation}
\ea
Although it is debatable whether this is really how the claims in~\cite{heisenberg27} should be interpreted and generalized, this relation is commonly attributed to Heisenberg in the literature~\cite{Ishikawa_1991,Ozawa03_physical_content,ozawa03,ozawa04,Ozawa04_gen_measurements,lund_wiseman,erhart12,rozema12_steinberg_expment}. Because it was actually first explicitly derived by Arthurs and Kelly~\cite{ArthursKelly65} (for position and momentum measurements---it was generalized to arbitrary observables by Arthurs and Goodman~\cite{ArthursGoodman88}), we will call it below the Heisenberg-Arthurs-Kelly relation.

This relation was indeed proven to hold, under some restrictive assumptions on the approximate joint measurements~\cite{ArthursKelly65,ArthursGoodman88,Ishikawa_1991,Ozawa_1991,ozawa04}; namely, it holds when ${\cal A}$ and ${\cal B}$ are such that the mean errors $\sandwich{\psi,\xi}{{\cal A} - A\!\otimes\! \one}{\psi,\xi}$ and $\sandwich{\psi,\xi}{{\cal B} - B\!\otimes\! \one}{\psi,\xi}$ are independent of the state $\ket{\psi}$. As we are only interested here in one particular state $\ket{\psi}$, for which we may want to adapt our approximation strategy, such an assumption is quite unsatisfactory for our purposes: we indeed aim at characterizing the trade-off between $\epsilon_{\cal A}$ and $\epsilon_{\cal B}$ for {\it all} possible approximate measurements, in which case {\it the Heisenberg-Arthurs-Kelly relation~\eqref{heisenberg_relation} does not generally hold}~\cite{ballentine}.

\subsection*{Ozawa's ``uncertainty relation''}

Only recently did Ozawa show~\cite{ozawa04} how one could derive a universally valid ``uncertainty relation'' for joint measurements, by adding two additional terms to the left-hand-side of Eq.~\eqref{heisenberg_relation}. His relation writes
\ba
\epsilon_{\cal A} \ \epsilon_{\cal B} + \Delta B \ \epsilon_{\cal A} + \Delta A \ \epsilon_{\cal B} & \geq & |C_{\!AB}| . \label{ozawa_relation}
\ea
(We note also that a very similar but inequivalent relation was derived by Hall~\cite{hall04}, which involves the standard deviations $\Delta {\cal A}$ and $\Delta {\cal B}$ rather than $\Delta A$ and $\Delta B$; see Supporting Information, Part~D for a discussion.)

The three terms in Ozawa's relation come from three independent uses of Robertson's relation~\eqref{robertson_UR} to different pairs of observables. While this indeed leads to a valid relation and allows one to exclude a large set of impossible values $(\epsilon_{\cal A}, \epsilon_{\cal B})$, this is not optimal, as the three Robertson's relations (and therefore Ozawa's relation) in general cannot be saturated simultaneously.

\subsection*{A new, tight error-trade-off relation for joint measurements}

Using a general geometric inequality for vectors in a Euclidean space (Lemma~\ref{ref_lemma1} of the Methods below), one can improve upon the sub-optimality of Ozawa's proof, and derive the following error-trade-off relation for approximate joint measurements:
\ba
\Delta B^2 \ \epsilon_{\cal A}^2 + \Delta A^2 \ \epsilon_{\cal B}^2 + 2 \sqrt{\Delta A^2 \Delta B^2 - C_{\!AB}^2} \ \epsilon_{\cal A} \, \epsilon_{\cal B} \ \geq \ C_{\!AB}^2 \, , \hspace{-5mm} \nonumber \\ \label{ETR}
\ea
or, in its dimensionless version---when $\Delta A, \Delta B \neq 0$, with $\tilde\epsilon_{\cal A} = \frac{\epsilon_{\cal A}}{\Delta A}$, $\tilde\epsilon_{\cal B} = \frac{\epsilon_{\cal B}}{\Delta B}$ and $\tilde{C}_{\!AB} = \frac{C_{\!AB}}{\Delta A\,\Delta B}$:
\ba
\tilde\epsilon_{\cal A}^2 + \tilde\epsilon_{\cal B}^2 + 2 \sqrt{1 - \tilde{C}_{\!AB}^2} \ \tilde\epsilon_{\cal A} \, \tilde\epsilon_{\cal B} \ \geq \ \tilde{C}_{\!AB}^2 . \label{ETR_dimensionless}
\ea
The proof is detailed in the Methods section.
It can easily be checked (see Supporting Information, Part~D) that Ozawa's relation~\eqref{ozawa_relation} can directly be derived from our new relation~\eqref{ETR}. Interestingly, one observes in particular that Ozawa's relation remains valid even if one drops the term $\epsilon_{\cal A} \, \epsilon_{\cal B}$---precisely the term that appears in the Heisenberg-Arthurs-Kelly relation~\eqref{heisenberg_relation}.

Not only is our relation stronger than Ozawa's, it is actually \emph{tight}: for any $A, B$ and $\ket{\psi}$, any values $(\epsilon_{\cal A}, \epsilon_{\cal B})$ saturating inequality~(\ref{ETR}--\ref{ETR_dimensionless}) can be obtained. This can even be achieved by projective measurements on $\ket{\psi}$, without introducing any ancillary system: see Supporting Information (Part~C) for explicit examples.
Hence, contrary to previously derived relations, our new one does not only tell what {\it cannot} be done quantum mechanically, but also what {\it can} be done.

Figure~\ref{fig_ETRs} illustrates the constraints imposed by the three error-trade-off relations~\eqref{heisenberg_relation}, \eqref{ozawa_relation} and~(\ref{ETR}--\ref{ETR_dimensionless}), in the plane $(\tilde\epsilon_{\cal A}, \tilde\epsilon_{\cal B})$. Our new relation~\eqref{ETR_dimensionless} thus characterizes precisely the optimal trade-off between $\tilde\epsilon_{\cal A}$ and $\tilde\epsilon_{\cal B}$ in the general context of approximate measurements. The values below the thick red curve cannot be reached, while all values on and above the curve can be obtained, by tuning the actual measurements ${\cal A}$ and ${\cal B}$ depending on how well one wants to measure one observable, at the expense of increasing the error on the other.

\begin{figure}
\begin{center}
\epsfxsize=7cm
\epsfbox{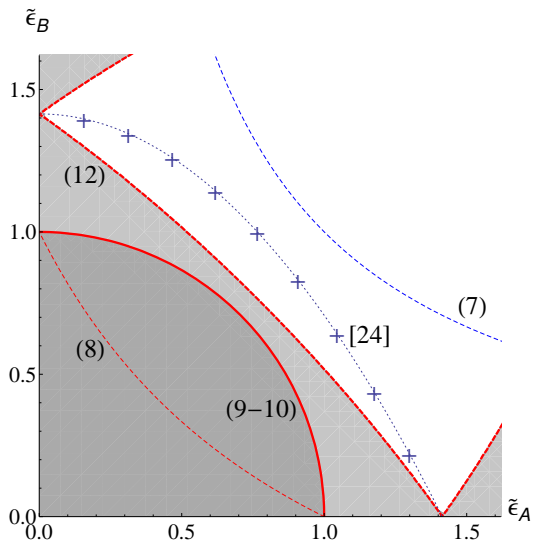}
\caption{{\bf Error-trade-off and error-disturbance relations.} The figure illustrates (in the case $\tilde C_{\!AB} = 1$) how the different error-trade-off and error-disturbance relations~\eqref{heisenberg_relation}, \eqref{ozawa_relation}, (\ref{ETR}--\ref{ETR_dimensionless}) and~\eqref{ETR_same_spectrum} restrict the possible values of the normalized rms errors $\tilde\epsilon_{\cal A}, \tilde\epsilon_{\cal B}$. Contrary to the Heisenberg-Arthurs-Kelly relation~\eqref{heisenberg_relation}, Ozawa's relation~\eqref{ozawa_relation} is always valid; however, it does not fully characterize the whole set of forbidden values for $(\tilde\epsilon_{\cal A}, \tilde\epsilon_{\cal B})$ (dark shaded area), which is precisely delimited by our new relation~(\ref{ETR}--\ref{ETR_dimensionless}). Imposing the same-spectrum assumption can imply strictly stronger constraints, such as Eq.~\eqref{ETR_same_spectrum} for the case where $A^2 = B^2 = \one$ and $\moy{A} = \moy{B} = 0$: more values of $(\tilde\epsilon_{\cal A}, \tilde\epsilon_{\cal B}) = (\epsilon_{\cal A}, \eta_{\cal B})$ are forbidden (light shaded area). \newline
The theoretical values expected from the experiment of Ref.~\cite{erhart12} are also shown; they do not saturate~\eqref{ETR_same_spectrum}, except for $\tilde\epsilon_{\cal A} = 0$ or $\tilde\epsilon_{\cal B} = 0$. On the other hand, in an ideal implementation the experiment of Ref.~\cite{rozema12_steinberg_expment} would saturate our inequality~\eqref{ETR_same_spectrum}.}
\label{fig_ETRs}
\end{center}
\end{figure}

\bigskip

\subsection*{The error-disturbance scenario \\ and the same-spectrum assumption}

Let us now consider a special case of our general framework for approximate joint measurements: that of the error-disturbance scenario, as first discussed by Heisenberg~\cite{heisenberg27}.

In this context, one considers the disturbance $\eta_{\cal B}$ in the statistics of one observable, $B$, due to the unsharp measurement of another observable, $A$. The latter is typically approximated by the measurement $M_A$ of a probe (or ancillary system, in the state $\ket{\xi}$), which interacts with the state $\ket{\psi}$ via a unitary transformation $U$~\cite{ozawa03}. In such a case, the approximation of $A$ corresponds to the measurement of ${\cal A} = U^\dagger (\one \otimes M_A) U$ on $\ket{\psi,\xi}$, while the perturbed measurement of $B$ after the interaction with the probe corresponds to the measurement of ${\cal B} = U^\dagger (B \otimes \one) U$ (note that ${\cal A}$ and ${\cal B}$ commute). This error-disturbance scenario can be cast into the same formalism as our joint measurement framework; the rms error $\epsilon_{\cal B}$ is now interpreted as the {\it rms disturbance} $\eta_{\cal B}$ of $B$, with formally the same definition~\cite{ozawa03}: $\eta_{\cal B} = \epsilon_{\cal B}$ as defined in~\eqref{def_epsB}.

Any error-trade-off relation derived in the more general framework of joint measurements thus remains valid in this error-disturbance scenario. In particular, when interpreting $\epsilon_{\cal B}$ as the rms disturbance $\eta_{\cal B}$, Ozawa's relation~\eqref{ozawa_relation} writes:
\ba
\epsilon_{\cal A} \ \eta_{\cal B} + \Delta B \ \epsilon_{\cal A} + \Delta A \ \eta_{\cal B} & \geq & |C_{\!AB}| . \label{ozawa_error_disturbance_relation}
\ea
This error-disturbance relation was actually introduced by Ozawa before its previous version~\eqref{ozawa_relation} for joint measurements~\cite{ozawa03}. In a similar manner, our new error-trade-off relation~(\ref{ETR}--\ref{ETR_dimensionless}) also implies a new error-disturbance relation, by simply replacing $\epsilon_{\cal B}$ by $\eta_{\cal B}$.

The difference with the previous, more general scenario of joint measurements is however not merely in the interpretation of $\epsilon_{\cal B}$.
A crucial point is that ${\cal B} = U^\dagger (B \otimes \one) U$ now has the same spectrum as $B$; it is furthermore typically (but often implicitly) assumed in the error-disturbance scenario that $M_A$, and hence ${\cal A} = U^\dagger (\one \otimes M_A) U$, also has the same spectrum as $A$~\cite{ozawa03,lund_wiseman,erhart12,rozema12_steinberg_expment}.
Because of these constraints, one may expect stronger restrictions on the possible values of $(\epsilon_{\cal A},\eta_{\cal B})$ to hold, and that stronger ``error-disturbance relations'' can be derived. (For simplicity, and by abuse of language, we call error-disturbance relation any error-trade-off relation derived under the same-spectrum assumption, as this is the crucial difference between the two scenarios.)

To illustrate this, let us now restrict our study to the case of dichotomic observables $A,B$ with eigenvalues $\pm 1$ (such that $A^2 = B^2 = \one$), and to states $\ket{\psi}$ for which $\moy{A} = \moy{B} = 0$ (which implies $\Delta A = \Delta B = 1$), as considered for instance in the experiments of Refs.~\cite{erhart12,rozema12_steinberg_expment}. We show in the Methods section that in this particular case, and with the same-spectrum assumption (hence, ${\cal A}^2 = {\cal B}^2 = \one$ as well), an analogous relation to~(\ref{ETR}--\ref{ETR_dimensionless}) holds, where $\epsilon_{\cal A}$ and $\epsilon_{\cal B}$ are replaced by $\epsilon_{\cal A} \sqrt{1-\frac{\epsilon_{\cal A}^2}{4}}$ and $\eta_{\cal B} \sqrt{1-\frac{\eta_{\cal B}^2}{4}}$, respectively:
\ba
&& \hspace{-.8cm} \epsilon_{\cal A}^2 \Big( 1-\frac{\epsilon_{\cal A}^2}{4} \Big) + \eta_{\cal B}^2 \Big( 1-\frac{\eta_{\cal B}^2}{4} \Big) \nonumber \\
&& \hspace{-.5cm} + \, 2 \sqrt{1 - C_{\!AB}^2} \ \, \epsilon_{\cal A} \sqrt{1-\frac{\epsilon_{\cal A}^2}{4}} \ \, \eta_{\cal B} \sqrt{1-\frac{\eta_{\cal B}^2}{4}} \ \geq \ C_{\!AB}^2 . \quad \label{ETR_same_spectrum}
\ea

This new error-disturbance relation is strictly stronger than~(\ref{ETR}--\ref{ETR_dimensionless}) (and than Ozawa's relation~\eqref{ozawa_error_disturbance_relation}). Furthermore, we show in the Supporting Information (Part~C) that it is tight when $|\sandwich{\psi}{AB}{\psi}|=1$: for any $A, B$ and $\ket{\psi}$ satisfying the constraints above, one can reach any values $(\epsilon_{\cal A},\eta_{\cal B})$ that saturate the inequality, using approximate measurements such that ${\cal A}^2 = {\cal B}^2 = \one$. The constraint that inequality~\eqref{ETR_same_spectrum} imposes on the possible values of $(\epsilon_{\cal A},\eta_{\cal B})$ is also illustrated on Figure~\ref{fig_ETRs}; note that contrary to our error-trade-off relation~(\ref{ETR}--\ref{ETR_dimensionless}), inequality~\eqref{ETR_same_spectrum} also bounds the possible values of $(\epsilon_{\cal A},\eta_{\cal B})$ from above (see also the inset of Figure~S1 in the Supporting Information).

Let us finally mention that if one imposes the same-spectrum assumption on ${\cal B}$ only (e.g., if one does not impose that $M_A$ in the specific error-disturbance scenario considered above has the same spectrum as $A$), one can also derive a similar, tight error-disturbance relation (under the assumptions now that $B^2 = \one$ and $\moy{B} = 0$), where only $\epsilon_{\cal B}$ in~(\ref{ETR}--\ref{ETR_dimensionless}) is replaced by $\eta_{\cal B} \sqrt{1-\frac{\eta_{\cal B}^2}{4}}$; see Eq.~(S20) in the Supporting Information (Part~B).

\subsection*{Example: qubits}

As an illustration of our error-trade-off and error-disturbance relations~(\ref{ETR}--\ref{ETR_dimensionless}) and~(\ref{ETR_same_spectrum}), let us consider the simplest case of qubits. We choose $\ket{\psi}$ to define the north pole of the Bloch sphere, and let $A = \hat {\mathbf a} \cdot \vec {\mathbf \sigma}$ and $B = \hat {\mathbf b} \cdot \vec {\mathbf \sigma}$ (where $\vec {\mathbf \sigma} = (\sigma_{\textsc{x}}, \sigma_{\textsc{y}}, \sigma_{\textsc{z}})$ denotes a vector composed of the 3 Pauli matrices) be two $\pm 1$-valued qubit observables characterized by unit vectors $\hat {\mathbf a}$ and $\hat {\mathbf b}$ on the Bloch sphere, of polar and azimuthal angles $\theta_a, \phi_a$ and $\theta_b, \phi_b$, respectively. We take $\theta_a, \theta_b \in [0,\pi]$ and we assume, for convenience, that $\phi = \phi_b - \phi_a \in [0, \frac{\pi}{2}]$.

\medskip

For such a choice of $\ket{\psi}$, $A$ and $B$, one finds $\Delta A = \sin\theta_a$, $\Delta B = \sin\theta_b$, and $C_{\!AB} = \sin\theta_a \sin\theta_b \sin\phi$. Equation~\eqref{ETR_dimensionless} then writes, for $\sin\theta_a \sin\theta_b \neq 0$:
\ba
\tilde\epsilon_{\cal A}^2 + \tilde\epsilon_{\cal B}^2 + 2 \cos\phi \ \tilde\epsilon_{\cal A} \, \tilde\epsilon_{\cal B} \ \geq \ \sin^2\!\phi .
\ea
One can check that this error-trade-off relation can simply be saturated by defining ${\cal A}$ and ${\cal B}$ to be projective measurements in the same eigenbasis, specified by any unit vector $\hat {\mathbf m}$ on the Bloch sphere with polar and azimuthal angles $\theta \in ]0,\pi[$ and $\varphi \in [\phi_a, \phi_b]$. More specifically, for ${\cal A} = [(\cos\theta_a{-}\cos\theta \, \hat {\mathbf m} \cdot \hat {\mathbf a}) \one{+}(\hat {\mathbf m} \cdot \hat {\mathbf a}{-}\cos\theta_a \cos\theta) \hat {\mathbf m} \cdot \vec {\mathbf \sigma}]/\sin^2\!\theta$ and ${\cal B} = [(\cos\theta_b{-}\cos\theta \, \hat {\mathbf m} \cdot \hat {\mathbf b}) \one{+}(\hat {\mathbf m} \cdot \hat {\mathbf b}{-}\cos\theta_b \cos\theta) \hat {\mathbf m} \cdot \vec {\mathbf \sigma}]/\sin^2\!\theta$, one obtains
\ba
\tilde\epsilon_{\cal A} \ = \ \sin (\varphi - \phi_a), \quad \tilde\epsilon_{\cal B} \ = \ \sin (\phi_b - \varphi).
\ea
Interestingly, $\tilde\epsilon_{\cal A}$ and $\tilde\epsilon_{\cal B}$ are independent of the polar angle $\theta$ of $\hat {\mathbf m}$. Note in particular that one can thus have for instance $\tilde\epsilon_{\cal A} = 0$ even when ${\cal A}$ is quite different from $A$; also, when $\hat {\mathbf m}$ comes close to the north or south poles of the Bloch sphere, one can have arbitrarily close projection directions $\hat {\mathbf m}$ leading to quite different values for $\tilde\epsilon_{\cal A}$ and $\tilde\epsilon_{\cal B}$---in our case, $(\tilde\epsilon_{\cal A}, \tilde\epsilon_{\cal B}) = (0,\sin\phi)$ for $\varphi=\phi_a$ and $(\tilde\epsilon_{\cal A}, \tilde\epsilon_{\cal B}) = (\sin\phi,0)$ for $\varphi=\phi_b$. These somewhat unexpected properties might however only be artefacts of the particular definitions of errors we use; it would be interesting to investigate possible alternative definitions that do not exhibit such behaviours.

\medskip

Let us now impose that ${\cal A}$ and ${\cal B}$ have the same spectrum as $A$ and $B$; i.e., since $A$ and $B$ are here $\pm 1$-valued observables, ${\cal A}^2 = A^2 = {\cal B}^2 = B^2 = \one$. Assuming furthermore that $\theta_a = \theta_b = \frac{\pi}{2}$, we have $\moy{A} = \moy{B} = 0$. Inequality~\eqref{ETR_same_spectrum} then applies; it can be saturated in the error-disturbance scenario [with ${\cal A} = U^\dagger (\one \otimes M_A) U$ and ${\cal B} = U^\dagger (B \otimes \one) U$] in the following way: let ${\mathbf \sigma}_\varphi = \cos\varphi \ \sigma_{\textsc{x}} + \sin\varphi \ \sigma_{\textsc{y}}$, for $\varphi \in [\phi_a,\phi_b]$ (i.e., ${\mathbf \sigma}_\varphi = \hat {\mathbf m} \cdot \vec {\mathbf \sigma}$ with $\theta = \frac{\pi}{2}$), and let $\ket{m_{\varphi}^\pm}$ be its (normalized) eigenvectors, corresponding to its eigenvalues $\pm 1$; we then define $M_A = {\mathbf \sigma}_\varphi$, $U = (U_{\textrm{R}}\!\otimes\!\one).U_{\textrm{copy}}$ with $U_{\textrm{copy}}$ a unitary such that $U_{\textrm{copy}} \ket{m_{\varphi}^+,\xi} = \ket{m_{\varphi}^+,m_{\varphi}^+}$ and $U_{\textrm{copy}} \ket{m_{\varphi}^-,\xi} = \ket{m_{\varphi}^-,m_{\varphi}^-}$ (e.g., with $\ket{\xi} = \ket{m_{\varphi}^+}$, a {\textsc{cnot}} unitary~\cite{NielsenChuang} in the $\{\ket{m_{\varphi}^\pm}\}$ basis), and $U_{\textrm{R}} = e^{-i\frac{\phi_b-\varphi}{2}\sigma_{\textsc{z}}}$; one then gets
\ba
\epsilon_{\cal A} \ = \ 2\,\sin\!\left(\!\frac{\varphi - \phi_a}{2}\!\right), \quad \eta_{\cal B} \ = \ 2\,\sin\!\left(\!\frac{\phi_b - \varphi}{2}\!\right) .
\ea

\subsubsection*{On the recent experimental tests of Refs.~\cite{erhart12} and~\cite{rozema12_steinberg_expment}}

Two experiments were recently reported, showing a violation of the Heisenberg-Arthurs-Kelly relation~\eqref{heisenberg_relation} (more specifically, of its error-disturbance version, where $\epsilon_{\cal B}$ is replaced by $\eta_{\cal B}$) and a verification of Ozawa's error-disturbance relation~\eqref{ozawa_error_disturbance_relation} in qubit systems.

The first experiment~\cite{erhart12} measured neutron spins, using the indirect method proposed in~\cite{Ozawa04_gen_measurements} to estimate the rms errors and rms disturbances $\epsilon_{\cal A}, \eta_{\cal B}$. $A = \sigma_{\textsc{x}}$ was estimated from the measurement of ${\mathbf \sigma}_\varphi = \cos\varphi \ \sigma_{\textsc{x}} + \sin\varphi \ \sigma_{\textsc{y}}$ on $\ket{\psi} = \ket{+\textsc{z}}$ (the eigenstate of $\sigma_{\textsc{z}}$, corresponding to its eigenvalue $+1$), and was followed by the measurement of $B = \sigma_{\textsc{y}}$; note that $A^2 = B^2 = \one$, $\moy{A} = \moy{B} = 0$, and $C_{\!AB} = 1$. The expected theoretical values for the rms errors and rms disturbances were $\epsilon_{\cal A} = 2 \sin \frac{\varphi}{2}$ and $\eta_{\cal B} = \sqrt{2} \cos\varphi$. These are plotted on Figure~\ref{fig_ETRs}; one can see that they are not optimal as they do not saturate our tight error-disturbance relation~\eqref{ETR_same_spectrum}. From the analysis above, it appears that adding a rotation $U_{\textrm{R}}$ before the measurement of $B$ would however be enough to allow the experimental setup used in~\cite{erhart12} to saturate inequality~\eqref{ETR_same_spectrum}.

The second experiment~\cite{rozema12_steinberg_expment} measured the polarization of single photons, using weak measurements as proposed in~\cite{lund_wiseman} to estimate the rms errors and rms disturbances.
$A$ was approximated from a measurement of variable strength based on a {\textsc{cnot}} unitary. Because the weak measurements used to estimate $\epsilon_{\cal A}$ and $\eta_{\cal B}$ are not infinitely weak, they slightly perturb the state of the photon, adding some noise. However, in an ideal implementation the experiment of~\cite{rozema12_steinberg_expment} would saturate the bound of our new error-disturbance relation~\eqref{ETR_same_spectrum}.

\medskip

To finish with, let us however emphasize that no experiment will ever {\it demonstrate} the universal validity of an ``uncertainty relation'' (or error-trade-off, or error-disturbance relations), despite what the title of Ref.~\cite{erhart12} suggests.
First note that in order for such experiments to be conclusive, one needs to perfectly trust the implementation; otherwise, systematic errors in the preparation of $\ket{\psi}$ or in the estimation procedure for $\epsilon_{\cal A}$ and $\epsilon_{\cal B}$ could radically change the values of the different terms in the relation, leading to unjustified conclusions (and possibly even ``showing'' a violation of a valid relation!). All one can do is then to check that in that particular (perfectly trusted) implementation, for some particular $A, B$ and $\ket{\psi}$ and for the particular approximations ${\cal A}$ and ${\cal B}$ implemented in that experiment, the error-trade-off or error-disturbance relation of interest is satisfied. There is indeed no way experimentally to test all possible approximate joint measurement strategies, and the particular choice of ${\cal A}$ and ${\cal B}$ could be non-optimal (as e.g. in~\cite{erhart12}).
It is of course trivial to obtain data satisfying an error-trade-off relation, if one does not try to optimize the values of $(\epsilon_{\cal A},\epsilon_{\cal B})$: if the relation is universally valid, then {\it any} measurement strategy (e.g. outputting random results) will satisfy it! One can even similarly trivially violate the Heisenberg-Arthurs-Kelly relation~\eqref{heisenberg_relation}, e.g. by actually measuring $A$ perfectly (so that $\epsilon_{\cal A} = 0$), and outputting any values to approximate $B$ (as long as $\epsilon_{\cal B} < \infty$). What is less trivial and therefore more interesting is to show experimentally that a tight error-trade-off or error-disturbance relation can indeed be {\it saturated}.

\section*{DISCUSSION}

We have presented a new, state-dependent error-trade-off relation [Eqs.~(\ref{ETR}--\ref{ETR_dimensionless}], in the general framework of approximate joint measurements. Our relation is universally valid, whether the Hilbert spaces of interest are of finite---as in our qubit example---or infinite dimensions (provided $\ket{\psi,\xi}$ is in the domains of $A^{(2)}, B^{(2)}, {\cal A}^{(2)}, {\cal B}^{(2)}$, and of all their products that are involved in the proof of~(\ref{ETR}--\ref{ETR_dimensionless}))---e.g. for the measurement of position and momentum, as first considered by Heisenberg. Note also that although the framework for joint measurements was presented for pure states, it can easily be generalized to mixed states, and Eqs.~(\ref{ETR}--\ref{ETR_dimensionless}) still hold.
Importantly, our new error-trade-off relation was shown to be \emph{tight}, and therefore to fully characterize the whole set of possible values of rms errors $(\epsilon_{\cal A},\epsilon_{\cal B})$ (in the case of pure states; our relation may in general not be tight for mixed states). This answers the question posed in the title, for pure states and when the quality of the measurement is quantified by these rms errors.

Error trade-off relations imply error-disturbance relations as a particular case. However, because of the same-spectrum assumption, strictly stronger relations can in general be derived in the error-disturbance scenario; we presented an example of such an error-disturbance relation, for $\pm1$-valued observables with $\moy{A} = \moy{B} = 0$, allowing us to highlight a quantitative difference between the two scenarios.
The derivation of a more general relation under the same-spectrum assumption is left for future work.

Our relations apply to the projective measurement of two observables $A$ and $B$. It would be interesting to see if these could be generalized to some POVMs (see~\cite{hall03_POM,hall04}, however, for the difficulties encountered), or to more observables~\cite{robertson2}. In the error-disturbance scenario, it may also be desirable to quantify the disturbance of the quantum state directly, rather than of the statistics of another observable; this is left as an open problem.

As highlighted above, our relations bound the rms errors of $A$ and $B$, as defined in Eqs.~(\ref{def_epsA}--\ref{def_epsB}). In the context of quantum information, one may however prefer to use information-theoretic definitions for the quality of approximations. Developing such definitions, and deriving corresponding universally valid and tight error-trade-off or error-disturbance relations would certainly be an interesting direction of research. This may indeed give a clearer operational meaning to such relations, and would be more adapted to their use in possible applications (in the same way e.g. as entropic uncertainty relations are useful to prove the security of quantum cryptographic protocols~\cite{WehnerWinter,Berta10}). This will involve radically different proof techniques, which may also allow one to consider error trade-offs in general probabilistic theories, not restricted to quantum theory and to its Hilbert space formalism. This will undoubtedly give more insight on the still puzzling, multi-faceted uncertainty principle.

\section*{METHODS}

In order to prove below our error-trade-off and error-disturbance relations~(\ref{ETR}--\ref{ETR_dimensionless}) and~(\ref{ETR_same_spectrum}), we start by introducing two general inequalities for real vectors.

\subsection*{Lemmas: Geometric inequalities}

Let $\hat a, \hat b$ be two unit vectors of a Euclidean space ${\cal E}$, and let us define $\chi = \hat a \cdot \hat b$. We prove in the Supporting Information (Part~B) the following lemmas:

\begin{lem}\label{ref_lemma1}
For any two \emph{orthogonal vectors} $\vec x$ and $\vec y$ of ${\cal E}$, one has
\ba
\| \hat a - \vec x \|^2 + \| \hat b - \vec y \|^2 + 2 \sqrt{1 - \chi^2} \ \| \hat a - \vec x \| \, \| \hat b - \vec y \| \ \geq \ \chi^2 \, . \hspace{-5mm} \nonumber \\ \label{ineq_lemma1}
\ea
\end{lem}

\begin{lem}\label{ref_lemma2}
For any two \emph{orthogonal unit vectors} $\hat x$ and $\hat y$ of ${\cal E}$, defining $a_\perp = \sqrt{1-(\hat a \cdot \hat x)^2}$ and $b_\perp = \sqrt{1-(\hat b \cdot \hat y)^2}$, one has
\ba
a_\perp^2 + b_\perp^2 + 2 \sqrt{1 - \chi^2} \ a_\perp \, b_\perp \ \geq \ \chi^2. \label{ineq_lemma2}
\ea
\end{lem}

\subsection*{Proof of our error-trade-off relation~(\ref{ETR}--\ref{ETR_dimensionless})}

Let us now define, in the non-trivial case $\Delta A \ \Delta B > 0$, the ket vectors
\ba
\ket{a} = \frac{A \!\otimes\! \one-\moy{A}}{\Delta A} \ket{\psi,\xi}, & \!\! & \ket{b} = \frac{B \!\otimes\! \one-\moy{B}}{\Delta B} \ket{\psi,\xi}, \quad \quad \label{def_keta_ketb} \\[1mm]
\ket{x} = \frac{{\cal A}-\moy{A}}{\Delta A} \ket{\psi,\xi}, & \!\! & \ket{y} = \frac{{\cal B}-\moy{B}}{\Delta B} \ket{\psi,\xi}. \label{def_ketx_kety}
\ea
By writing these vectors in any orthonormal basis of ${\cal H} \otimes {\cal K}$ (e.g., the common eigenbasis of ${\cal A}$ and ${\cal B}$), and denoting by ${\mathrm{Re}}$ and ${\mathrm{Im}}$ their real and imaginary parts, respectively, one can define the following real vectors:
\ba
\hat a = \left(
\begin{array}{c}
{\mathrm{Re}}\ket{a} \\
{\mathrm{Im}}\ket{a}
\end{array}
\right), \
\hat b = \left(
\begin{array}{c}
{\mathrm{Im}}\ket{b} \\
-{\mathrm{Re}}\ket{b}
\end{array}
\right), \label{def_veca_vecb} \\[1mm]
\vec x = \left(
\begin{array}{c}
{\mathrm{Re}}\ket{x} \\
{\mathrm{Im}}\ket{x}
\end{array}
\right), \
\vec y = \left(
\begin{array}{c}
{\mathrm{Im}}\ket{y} \\
-{\mathrm{Re}}\ket{y}
\end{array}
\right). \label{def_vecx_vecy}
\ea

\begin{figure}
\begin{center}
\epsfxsize=7.5cm
\epsfbox{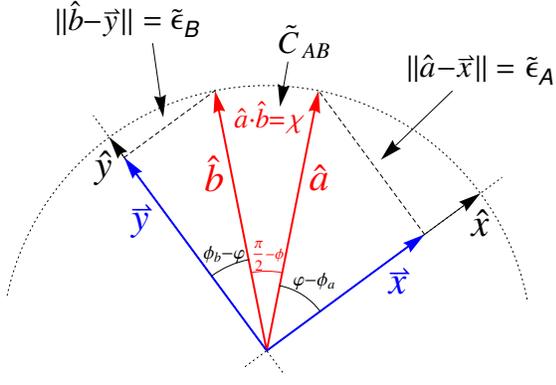}
\caption{{\bf Geometric construction used in the proof of our general error-trade-off relation~(\ref{ETR}--\ref{ETR_dimensionless}).} The real vectors $\hat a$, $\hat b$, $\vec x$ and $\vec y$ satisfy the assumptions of Lemma~\ref{ref_lemma1}; the particular choice of vectors illustrated here (for which $\chi = \sin\phi = \sin(\phi_b{-}\phi_a)$, $||\hat a - \vec x|| = \sin(\varphi - \phi_a)$ and $||\hat b - \vec y|| = \sin(\phi_b - \varphi)$) saturates inequality~\eqref{ineq_lemma1}, which quantifies the optimal trade-off between the distance from the unit vector $\hat a$ to an axis along a direction $\hat x$, and from the unit vector $\hat b$ to an axis along a direction $\hat y$, orthogonal to $\hat x$.}
\label{fig_proof_ETR}
\end{center}
\end{figure}

One then has
\ba
\|\hat a\|^2 &\!=\!& ({\mathrm{Re}}\ket{a})^{\!\top} \!\!\cdot\! ({\mathrm{Re}}\ket{a}) + ({\mathrm{Im}}\ket{a})^{\!\top} \!\!\cdot\! ({\mathrm{Im}}\ket{a}) = \braket{a}{a} = 1, \nonumber \\[-1mm] \label{a2_sigA2} \\
\|\hat b\|^2 &\!=\!& 1,
\ea
\ba
\|\vec x - \hat a\|^2 &\!=\!& (\bra{x}-\bra{a})(\ket{x}-\ket{a}) \nonumber \\
 &\!=\!& \bra{\psi,\xi}\Big(\frac{{\cal A}{-}A \otimes \one}{\Delta A}\Big)^{\!2}\ket{\psi,\xi} = \frac{\epsilon_{\cal A}^2}{\Delta A^2} = \tilde\epsilon_{\cal A}^2, \qquad \ \label{epsA_dist_x_a} \\
\|\vec y - \hat b\|^2 &\!=\!& \frac{\epsilon_{\cal B}^2}{\Delta B^2} \, = \, \tilde\epsilon_{\cal B}^2, \label{epsB_dist_y_b}
\ea
\vspace{-5mm}
\ba
\hat a \cdot \hat b &=& ({\mathrm{Re}}\ket{a})^{\!\top} \!\!\cdot ({\mathrm{Im}}\ket{b}) - ({\mathrm{Im}}\ket{a})^{\!\top} \!\!\cdot ({\mathrm{Re}}\ket{b}) = {\mathrm{Im}}\,\braket{a}{b} \nonumber \\
 &=& \frac{1}{2i} \frac{\sandwich{\psi}{[A,B]}{\psi}}{\Delta A \, \Delta B} = \frac{C_{\!AB}}{\Delta A \, \Delta B} \, = \, \tilde C_{\!AB}, \qquad \label{ab_CAB} \\[2mm]
\vec x \cdot \vec y &=& \frac{1}{2i} \frac{\sandwich{\psi,\xi}{[{\cal A}, {\cal B}]}{\psi,\xi}}{\Delta A \, \Delta B} = 0.
\ea
Hence, the (normalized) rms errors $\tilde\epsilon_{\cal A}$, $\tilde\epsilon_{\cal B}$ can be interpreted as distances between vectors~\cite{hall03_POM,Ozawa03_physical_content}, while the commutativity of ${\cal A}$ and ${\cal B}$ translates into an orthogonality condition for $\vec x$ and $\vec y$.

The vectors $\hat a$, $\hat b$, $\vec x$ and $\vec y$ thus satisfy the assumptions of Lemma~\ref{ref_lemma1}, that $||\hat a|| = ||\hat b|| = 1$ and $\vec x \cdot \vec y = 0$ (see Figure~\ref{fig_proof_ETR}). Together with Eqs.~(\ref{epsA_dist_x_a}--\ref{ab_CAB}), inequality~\eqref{ineq_lemma1} implies our general error-trade-off relation for joint measurements~\eqref{ETR_dimensionless}. After multiplication by $\Delta A^2 \Delta B^2$, we obtain Eq.~\eqref{ETR} (for which the case $\Delta A \, \Delta B = 0$ is trivial, as it implies $C_{\!AB} = 0$).

\subsection*{Proof of our error-disturbance relation~(\ref{ETR_same_spectrum}), for the case where $A^2 = B^2 = {\cal A}^2 = {\cal B}^2 = \one$ and $\moy{A} = \moy{B} = 0$}

\begin{figure}
\begin{center}
\epsfxsize=7.5cm
\epsfbox{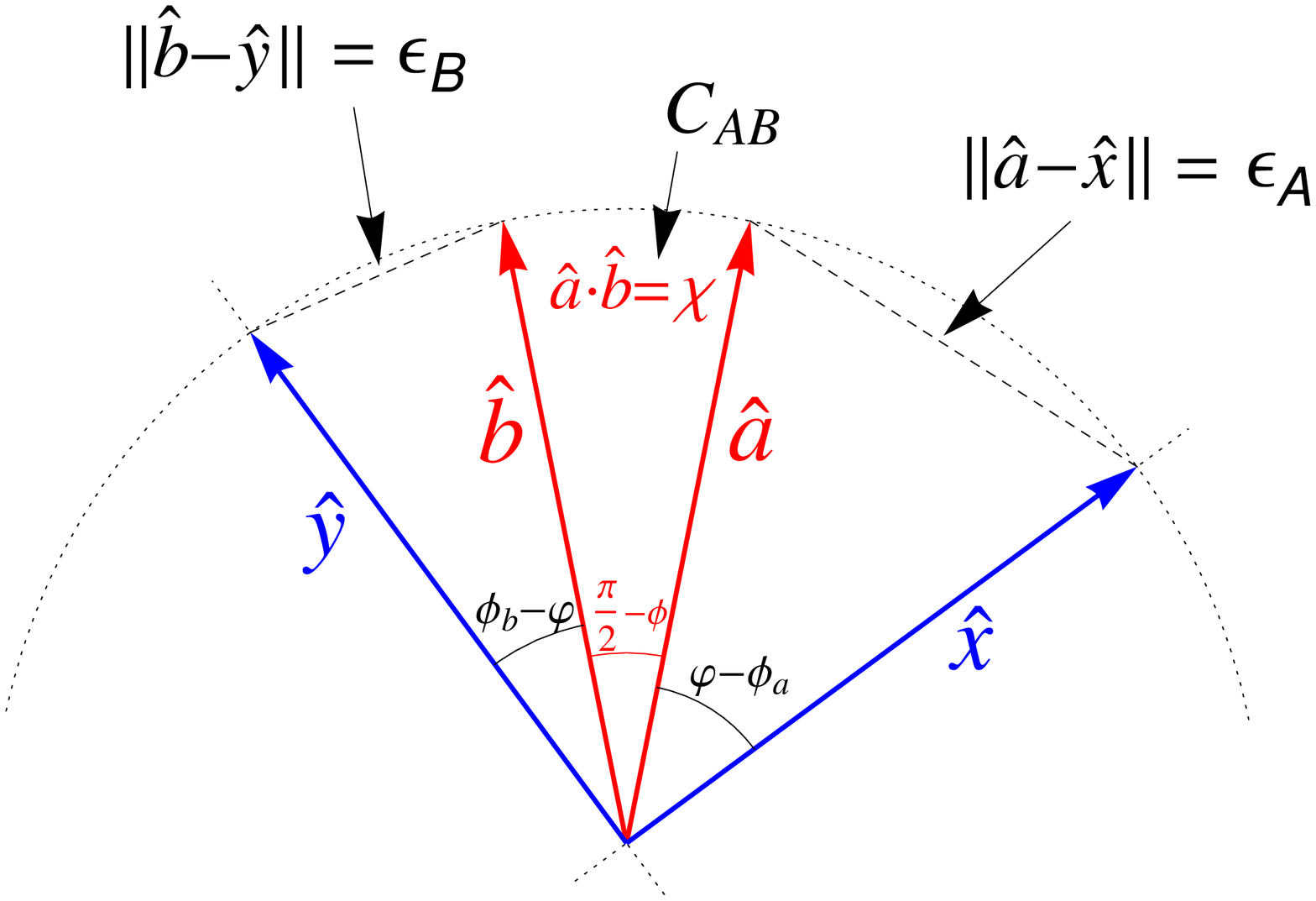}
\caption{{\bf Geometric construction used in the proof of our error-disturbance relation~(\ref{ETR_same_spectrum})} (for the case where $A^2 = B^2 = {\cal A}^2 = {\cal B}^2 = \one$ and $\moy{A} = \moy{B} = 0$). The real vectors $\hat a$, $\hat b$, $\hat x$ and $\hat y$ satisfy the assumptions of Lemma~\ref{ref_lemma2}; the particular choice of vectors illustrated here (for which $\chi = \sin\phi = \sin(\phi_b{-}\phi_a)$, $||\hat a - \hat x|| = 2\sin\frac{\varphi - \phi_a}{2}$ and $||\hat b - \hat y|| = 2\sin\frac{\phi_b-\varphi}{2}\ $) saturates inequality~\eqref{ineq_lemma2}, which quantifies the optimal trade-off between the distance from the unit vector $\hat a$ to another unit vector $\hat x$, and from the unit vector $\hat b$ to another unit vector $\hat y$, orthogonal to $\hat x$.}
\label{fig_proof_ETR_same_spectrum}
\end{center}
\end{figure}

With the assumptions that $A^2 = B^2 = \one$ and $\moy{A} = \moy{B} = 0$ (hence $\Delta A = \Delta B = 1$), and that ${\cal A}$ and ${\cal B}$ have the same spectrum as $A$ and $B$ (hence ${\cal A}^2 = {\cal B}^2 = \one$), the real vectors $\hat a, \hat b, \hat x \,(=\!\vec x)$ and $\hat y \,(=\!\vec y)$ defined as in~(\ref{def_veca_vecb}--\ref{def_vecx_vecy}) are now such that (with $\epsilon_{\cal B} = \eta_{\cal B}$ in the error-disturbance scenario)
\ba
&& \hspace{-3mm} \|\hat a\| \,=\, \|\hat b\| \,=\, \|\hat x\| \,=\, \|\hat y\| \,=\, 1, \\[2mm]
&& \hspace{-3mm} a_\perp^2 =\, 1 - (\hat a \!\cdot\! \hat x)^2 = \|\hat x{-}\hat a\|^2 \Big(1\!-\!\frac{\|\hat x{-}\hat a\|^2\!}{4} \Big) = \epsilon_{\cal A}^2 \Big(1 - \frac{\epsilon_{\cal A}^2}{4} \Big) \nonumber \\ \label{epsA_a_perp} \\[-1mm]
&& \hspace{-3mm} b_\perp^2 \,=\, 1 - (\hat b \cdot \hat y)^2 \,=\, \eta_{\cal B}^2 \Big(1 - \frac{\eta_{\cal B}^2}{4} \Big), \label{etaB_b_perp} \\[1mm]
&& \hspace{-3mm} \hat a \cdot \hat b \,=\, C_{\!AB}, \qquad \hat x \cdot \hat y \,=\, 0. \label{ab_CAB_bis}
\ea

The vectors $\hat a$, $\hat b$, $\hat x$ and $\hat y$ thus satisfy the assumptions of Lemma~\ref{ref_lemma2} (see Figure~\ref{fig_proof_ETR_same_spectrum}). Together Eqs.~(\ref{epsA_a_perp}--\ref{ab_CAB_bis}), inequality~\eqref{ineq_lemma2} gives our error-disturbance relation~\eqref{ETR_same_spectrum}.

\section*{ACKNOWLEDGMENTS}

I am grateful to M.~J.~W. Hall for fruitful discussions and comments on an earlier version of this manuscript. This work was supported by a UQ Postdoctoral Research Fellowship.

\onecolumngrid
\clearpage

\renewcommand{\theequation}{S\arabic{equation}}
\renewcommand{\thefigure}{S\arabic{figure}}

\setcounter{equation}{0}
\setcounter{figure}{0}

\renewcommand\thesection{\Alph{section}} 
\renewcommand\thesubsection{\arabic{subsection}}

$ $
\bigskip

\begin{center}
{\bf \large SUPPORTING INFORMATION}
\end{center}

\bigskip

\begin{adjustwidth}{1.95cm}{1.95cm}
{\small
\hspace{2mm} This Supporting Information starts by giving more details on our general framework for approximate joint measurements (Part~\ref{secSI_approx_joint_meas}). We then prove the two geometric Lemmas used in the proofs of our error-trade-off and error-disturbance relations (Part~\ref{SecSI_proof_lemmas}), and prove the tightness of our relations (Part~\ref{secSI_proof_tightness}). We finally show explicitly how Ozawa's ``uncertainty relation'' for joint measurements follows from our new error-trade-off relation (Part~\ref{SecSI_ozawa_from_new_relation}).
}
\end{adjustwidth}

\bigskip
\bigskip

\twocolumngrid

\section{More on approximate joint measurements}
\label{secSI_approx_joint_meas}

In the main text we have presented our general framework for the approximate joint measurement of $A$ and $B$ on $\ket{\psi}$ by introducing two commuting observables ${\cal A}$ and ${\cal B}$ to be measured on the state $\ket{\psi,\xi}$. We give below an alternative description in terms of Positive Operator-Valued Measures (POVMs)~\cite{NielsenChuang_SI}, which justifies the generality of our framework; we then discuss how to optimize the quality of the approximations, for a given POVM or a given projection eigenbasis.

\subsection{General strategy for approximate joint measurements}

In quantum theory, the most general strategy to get some information on a quantum system in the state $\ket{\psi} \in {\cal H}$ is to perform a POVM; we denote such a general\footnote{Note that we do note make any restrictive assumption on the possible POVM outcomes: there can be an arbitrary number of different real values $m$, and these could even be continuous, in which case sums would need to be appropriately replaced by integrals.} POVM by $\mathbb{M} = \{{\text M}_m\}$, where its elements ${\text M}_m$ are Hermitian, positive operators acting on ${\cal H}$, with $\sum_m {\text M}_m = \one$. In order to approximate the measurements of two observables $A$ and $B$ on $\ket{\psi}$, one can then use the results of the POVM $\mathbb{M}$ in the following way: for each of its possible outcomes $m$, we define some real value $f(m)$ that aims at estimating the result of an ideal measurement of $A$, and some real value $g(m)$ that aims at estimating the result of an ideal measurement of $B$ (where without loss of generality, $f(m)$ and $g(m)$ can be assumed to be deterministic functions).

Now, from Neumark's extension theorem~\cite{peres1990_SI}, any POVM can equivalently be represented as a projective measurement on an extended Hilbert space.
For convenience, we can thus assume that $\mathbb{M}$ corresponds to the measurement of a Hermitian observable ${\cal M} = \sum_m m \, \ket{m}\!\bra{m}$ on the state $\ket{\psi,\xi} = \ket{\psi} \otimes \ket{\xi}$, where $\ket{\xi} \in {\cal K}$ is the state of an ancillary system. In this picture, the projection of $\ket{\psi,\xi}$ onto a (normalized) eigenvector $\ket{m} \in {\cal H} \otimes {\cal K}$ gives the eigenvalue $m$ of ${\cal M}$, corresponding to the outcome $m$ of the POVM $\mathbb{M}$.
The approximation strategy for $A$ and $B$ described above, using the output values $f(m)$ to approximate the measurement of $A$ and $g(m)$ for that of $B$, then actually corresponds to the (joint) measurement of the observables ${\cal A} = \sum_m f(m) \, \ket{m}\!\bra{m} = f({\cal M})$ and ${\cal B} = \sum_m g(m) \, \ket{m}\!\bra{m} = g({\cal M})$ on $\ket{\psi,\xi}$.
These two commuting observables correspond precisely to the observables ${\cal A}$ and ${\cal B}$ considered in the main text.

\subsection{Quality of the approximations}

To characterize the quality of the approximations of $A$ and $B$, we have introduced in the main text the rms errors such that
\ba
\epsilon_{\cal A}^2 &=& \sandwich{\psi,\xi}{\,({\cal A} - A \otimes \one)^2\,}{\psi,\xi} , \label{eqSI_def_epsA} \\
\epsilon_{\cal B}^2 &=& \sandwich{\psi,\xi}{\,({\cal B} - B \otimes \one)^2\,}{\psi,\xi} . \label{eqSI_def_epsB}
\ea
It is worth noting that $\epsilon_{\cal A}$ and $\epsilon_{\cal B}$ do not depend on the particular, non-unique Neumark extension (specified by the Hilbert space ${\cal K}$, $\ket{\xi} \in {\cal K}$ and the eigenvectors $\ket{m} \in {\cal H} \!\otimes\! {\cal K}$) chosen for the POVM $\mathbb{M}$. Using the definitions ${\cal A} = \sum f(m) \ket{m}\!\bra{m}$ and ${\cal B} = \!\sum g(m) \ket{m}\!\bra{m}$, and the fact that for all $m$, ${\text M}_m = \one \!\otimes\! \bra{\xi} \, . \, \ket{m}\!\bra{m} \, . \, \one \! \otimes \! \ket{\xi}$, the rms errors can indeed be directly expressed in terms of the POVM elements ${\text M}_m$ and of the functions $f(m)$ and $g(m)$ as follows~\cite{hall04_SI,ozawa04_SI}:
\ba
\epsilon_{\cal A}^2 &=& \sum_m \ \sandwich{\psi}{\big(A-f(m)\big){\text M}_m\big(A-f(m)\big)}{\psi} , \label{eqSI_epsA_POVM} \\
\epsilon_{\cal B}^2 &=& \sum_m \ \sandwich{\psi}{\big(B-g(m)\big){\text M}_m\big(B-g(m)\big)}{\psi} . \label{eqSI_epsB_POVM}
\ea

\subsection{Optimal choice for $f(m)$ and $g(m)$}

For a given POVM $\mathbb{M}$ or a given projective measurement of $\ket{\psi,\xi}$ onto an eigenbasis $\{\ket{m}\}$, one can choose the output functions $f(m)$ and $g(m)$ so as to optimize the quality of our approximations (i.e., minimize $\epsilon_{\cal A}$ and $\epsilon_{\cal B}$).

Developing Eq.~\eqref{eqSI_def_epsA} with ${\cal A} = \sum_m f(m) \, \ket{m}\!\bra{m}$, one indeed finds\footnote{In the following, the notation $\moy{...}$ is used for $\sandwich{\psi}{...}{\psi}$. The expression inside the brackets does not need to be a Hermitian observable; if it is, the notation thus denotes its mean value in the state $\ket{\psi}$.}~\cite{hall04_SI}
\ba
\epsilon_{\cal A}^2 &=& \moy{A^2} \ - \!\! \sum_{m | p(m) > 0} p(m) \, \left( \mathrm{Re} \, \frac{\sandwich{m}{A\!\otimes\!\one \,}{\psi,\xi}}{\braket{m}{\psi,\xi}} \right)^{\!2} \nonumber \\[-1mm]
 && \ + \!\!\! \sum_{m | p(m) > 0} \!\!\! p(m) \left( f(m) - \mathrm{Re} \, \frac{\sandwich{m}{A\!\otimes\!\one \,}{\psi,\xi}}{\braket{m}{\psi,\xi}} \right)^{\!2} \!\!\! , \qquad \label{eqSI_epsA2_to_optimize}
\ea
with $p(m) = |\braket{m}{\psi,\xi}|^2$ (and where the sums are over the normalized eigenvectors $\ket{m}$ of ${\cal M}$ for which $p(m) > 0$). The value $f(m)$ only contributes to the last sum in~\eqref{eqSI_epsA2_to_optimize}. It thus appears clearly that in order to minimize the rms errors $\epsilon_{\cal A}$ and $\epsilon_{\cal B}$, the optimal values for $f(m)$ and similarly for $g(m)$ are, for each possible outcome $m$,
\ba
f_{opt}(m) = \mathrm{Re} \, \frac{\sandwich{m}{A\!\otimes\!\one \,}{\psi,\xi}}{\braket{m}{\psi,\xi}}, \label{eqSI_f_opt} \\
g_{opt}(m) = \mathrm{Re} \, \frac{\sandwich{m}{B\!\otimes\!\one \,}{\psi,\xi}}{\braket{m}{\psi,\xi}}. \label{eqSI_g_opt}
\ea
Interestingly, one may recognize above that the optimal values $f_{opt}(m)$ and $g_{opt}(m)$ are given by the real parts of the so-called weak values~\cite{WeakMeas_AAV_SI} of the observables $A\!\otimes\!\one$ and $B\!\otimes\!\one$, respectively, pre-selected in the state $\ket{\psi,\xi}$ and post-selected in the state $\ket{m}$. This thus provides an interesting interpretation to the real part of weak values as the optimal approximation of an observable, when the criterion for optimality is taken to be the rms error~\cite{hall04_SI}.

For this optimal choice of output functions $f = f_{opt}$ and $g = g_{opt}$, one then gets, from~\eqref{eqSI_epsA2_to_optimize} and from the fact that $\moy{A^2} = \sum_m \sandwich{\psi,\xi}{A\!\otimes\!\one}{m} \sandwich{m}{A\!\otimes\!\one}{\psi,\xi}$, where the sum now runs over all normalized eigenstates $\ket{m}$ of ${\cal M}$:
\ba
\epsilon_{\cal A}^2(f_{opt}) &=& \sum_{m | p(m) > 0} p(m) \, \left( \mathrm{Im} \, \frac{\sandwich{m}{A\!\otimes\!\one \,}{\psi,\xi}}{\braket{m}{\psi,\xi}} \right)^2 \nonumber \\
&& \ + \sum_{m | p(m) = 0} |\sandwich{m}{A\!\otimes\!\one}{\psi,\xi}|^2 \, \!\!, \label{epsA2_optim} \\
\epsilon_{\cal B}^2(g_{opt}) &=& \sum_{m | p(m) > 0} p(m) \, \left( \mathrm{Im} \, \frac{\sandwich{m}{B\!\otimes\!\one \,}{\psi,\xi}}{\braket{m}{\psi,\xi}} \right)^2 \nonumber \\
&& \ + \sum_{m | p(m) = 0} |\sandwich{m}{B\!\otimes\!\one}{\psi,\xi}|^2 \, \!\!, \label{epsB2_optim}
\ea
which now involve the imaginary parts of the weak values~\cite{Johansen2004_SI}, together with some terms related to eigenvalues $m$ for which $p(m) = 0$.

\medskip

Note again that $f_{opt}(m)$ and $g_{opt}(m)$ can be expressed in terms of the POVM elements ${\text M}_m$ directly. Namely:
\ba
f_{opt}(m) = \mathrm{Re} \frac{\moy{{\text M}_m A}}{\moy{{\text M}_m}}, \quad
g_{opt}(m) = \mathrm{Re} \frac{\moy{{\text M}_m B}}{\moy{{\text M}_m}}, \quad \label{eqSI_fgopt_POVM}
\ea
where $\moy{{\text M}_m A}/\moy{{\text M}_m}$ and $\moy{{\text M}_m B}/\moy{{\text M}_m}$ are the standard generalisations of weak values for POVMs~\cite{wiseman_general_weak_values_SI}. However, in general one has (with now $p(m) = \moy{{\text M}_m}$):
\ba
\epsilon_{\cal A}^2(f_{opt}) &=& \moy{A^2} \ - \! \sum_{m | p(m) > 0} p(m) \, [ \mathrm{Re} \, \moy{{\text M}_m A}/\moy{{\text M}_m} ]^2 \nonumber \\
=&& \hspace{-4mm} \! \sum_{m | p(m) > 0} \!\!\!\! p(m) \! \left( \mathrm{Im} \, \frac{\moy{{\text M}_m A}}{\moy{{\text M}_m}} \right)^2 + \!\!\! \sum_{m | p(m) = 0} \!\!\!\! \moy{A M_m A} \hspace{-2mm} \nonumber \\
&& \quad + \sum_{m | p(m) > 0} \!\! \frac{\moy{M_m}\moy{A M_m A} - |\moy{M_m A}|^2}{p(m)} \nonumber \\
\geq&& \hspace{-4mm} \! \sum_{m | p(m) > 0} \!\!\!\! p(m) \! \left( \mathrm{Im} \, \frac{\moy{{\text M}_m A}}{\moy{{\text M}_m}} \right)^2 + \!\!\! \sum_{m | p(m) = 0} \!\!\!\! \moy{A M_m A}, \hspace{-2mm} \nonumber \\[-1mm] \\[4mm]
\epsilon_{\cal B}^2(f_{opt}) &=& \moy{B^2} \ - \! \sum_{m | p(m) > 0} p(m) \, [ \mathrm{Re} \, \moy{{\text M}_m B}/\moy{{\text M}_m} ]^2 \nonumber \\
\geq&& \hspace{-4mm} \! \sum_{m | p(m) > 0} \!\!\!\! p(m) \! \left( \mathrm{Im} \, \frac{\moy{{\text M}_m B}}{\moy{{\text M}_m}} \right)^2 + \!\!\! \sum_{m | p(m) = 0} \!\!\!\! \moy{B M_m B}, \hspace{-2mm} \nonumber \\[-1mm]
\ea
where the inequalities follow from the Cauchy-Schwarz inequality applied e.g. to $M_m^{1/2} \ket{\psi}$ and $M_m^{1/2} A \ket{\psi}$.

\medskip

Finally, one can easily check that the optimal values $f_{opt}(m)$ and $g_{opt}(m)$ derived above are such that the approximations of $A$ and $B$ are ``unbiased'', in the sense that $\moy{{\cal A}} = \moy{A}$ and $\moy{{\cal B}} = \moy{B}$.
Let us emphasize however that $f_{opt}(m)$ and $g_{opt}(m)$ may not be in the spectrum of $A$ and $B$. While for general joint measurements we indeed do not require the approximation functions to output eigenvalues of $A$ and $B$, we nevertheless impose such a constraint in the error-disturbance scenario (see main text). In that case, it may not be possible for $f(m)$ and $g(m)$ to take the values $f_{opt}(m)$ and $g_{opt}(m)$ given by Eqs.~(\ref{eqSI_f_opt}--\ref{eqSI_g_opt}) or~\eqref{eqSI_fgopt_POVM}; see Part~\ref{secSI_proof_tightness2_ETR_same_spectrum} below for an example.

\section{Proofs of our geometric Lemmas}
\label{SecSI_proof_lemmas}

We now prove the geometric inequalities of the Methods section, which hold for any two unit vectors $\hat a, \hat b$ of a Euclidean space ${\cal E}$ (with $\hat a \cdot \hat b = \chi$). Let us start with Lemma~\ref{ref_lemma2}:

\setcounter{thm}{1}
\begin{lem}
For any two \emph{orthogonal unit vectors} $\hat x$ and $\hat y$ of ${\cal E}$, defining $a_\perp = \sqrt{1-(\hat a \cdot \hat x)^2}$ and $b_\perp = \sqrt{1-(\hat b \cdot \hat y)^2}$, one has
\ba
a_\perp^2 + b_\perp^2 + 2 \sqrt{1 - \chi^2} \ a_\perp \, b_\perp \ \geq \ \chi^2. \label{ineq_lemma2_SI}
\ea
\end{lem}

\medskip

\begin{proof}[Proof of Lemma~\ref{ref_lemma2}]

We use for convenience the notation $a_x = \hat a \cdot \hat x$.
Let us define $\hat a_\perp^{\, (b)} = \frac{\hat b - \chi \hat a}{\sqrt{1-\chi^2}}$ if $|\chi| \neq 1$, $\hat a_\perp^{\, (b)} = \hat 0$ otherwise, and $\hat a_\perp^{\, (x)} = \frac{\hat x - a_x \hat a}{a_\perp}$ if $a_\perp \neq 0$, $\hat a_\perp^{\, (x)} = \hat 0$ otherwise, so that $\hat b = \chi \hat a + \sqrt{1-\chi^2} \hat a_\perp^{\, (b)}$ and $\hat x = a_x \hat a + a_\perp \hat a_\perp^{\, (x)}$, with $\|\hat a_\perp^{\, (b)}\| = 1$ or 0, $\|\hat a_\perp^{\, (x)}\| = 1$ or 0, and $\hat a \cdot \hat a_\perp^{\, (b)} = \hat a \cdot \hat a_\perp^{\, (x)} = 0$.

Since $\hat x$ and $\hat y$ are orthogonal unit vectors, we have $(\hat b \cdot \hat x)^2 + (\hat b \cdot \hat y)^2 \leq \|\hat b\|^2 = 1$, and hence
\ba
b_\perp &=& \sqrt{1-(\hat b \cdot \hat y)^2} \nonumber \\
& \geq & |\hat b \cdot \hat x| \ = \ \left| \big(\chi \hat a + \sqrt{1-\chi^2} \hat a_\perp^{\, (b)} \big) \cdot \big(a_x \hat a + a_\perp \hat a_\perp^{\, (x)} \big) \right| \nonumber \\
& & \phantom{|\hat b \cdot \hat x| \ } = \ \left| \chi a_x + \sqrt{1-\chi^2} \ a_\perp \ \hat a_\perp^{\, (b)} \! \cdot \hat a_\perp^{\, (x)} \right| \nonumber \\
& & \phantom{|\hat b \cdot \hat x| \ } \geq \ | \chi a_x | - \sqrt{1-\chi^2} \ a_\perp \, . \label{eq_proof_lemma1}
\ea
It follows that
\ba
\left( b_\perp + \sqrt{1-\chi^2} \ a_\perp \right)^2 \ \geq \ \left( \chi a_x \right)^2 \ = \ \chi^2 \, (1 - a_\perp^2) , \quad
\ea
which is equivalent to~\eqref{ineq_lemma2_SI}.

Note that this inequality can only be saturated if $\hat a, \hat b, \hat x$ and $\hat y$ are coplanar: indeed, the first inequality in~\eqref{eq_proof_lemma1} would need to be an equality (which requires $\hat b \in {\mathrm{Span}}\{\hat x, \hat y\}$), and a similar constraint would apply to $\hat a$.
\end{proof}

\medskip

We can now use the result of Lemma~\ref{ref_lemma2} to prove Lemma~\ref{ref_lemma1}:

\setcounter{thm}{0}
\begin{lem}
For any two \emph{orthogonal vectors} $\vec x$ and $\vec y$ of ${\cal E}$, one has
\ba
\| \hat a - \vec x \|^2 + \| \hat b - \vec y \|^2 + 2 \sqrt{1 - \chi^2} \ \| \hat a - \vec x \| \, \| \hat b - \vec y \| \ \geq \ \chi^2 \, . \hspace{-5mm} \nonumber \\ \label{ineq_lemma1_SI}
\ea
\end{lem}

\begin{proof}[Proof of Lemma~\ref{ref_lemma1}]

Let us define $\hat x = \frac{\vec x}{\| \vec x \|}$ (if $\| \vec x \| \neq 0$; otherwise, $\hat x$ is defined as any unit vector orthogonal to $\vec y$), $\hat y = \frac{\vec y}{\| \vec y \|}$ (if $\| \vec y \| \neq 0$; otherwise, $\hat y$ is defined as any unit vector orthogonal to $\hat x$), and let us use the notations of Lemma~\ref{ref_lemma2}: $a_x = \hat a \cdot \hat x$, $a_\perp = \sqrt{1-a_x^2}$, and $b_\perp = \sqrt{1-(\hat b \cdot \hat y)^2}$.
One has
\ba
\|\hat a - \vec x\|^2 &=& \left\| \big( \hat a - a_x \hat x \big) + \big( a_x \hat x - \vec x \big) \right\|^2 \nonumber \\
&=& \big\| \hat a - a_x \hat x \big\|^2 + \big\| a_x \hat x - \vec x \big\|^2 \nonumber \\
& \geq & \big\| \hat a - a_x \hat x \big\|^2 = 1 - a_x^2 = a_\perp^2 , \label{eq_proof_lemma1_SM}
\ea
and similarly, $\|\hat b - \vec y\|^2 \geq b_\perp^2$. Therefore,
\ba
\| \hat a - \vec x \|^2 + \| \hat b - \vec y \|^2 + 2 \sqrt{1 - \chi^2} \ \| \hat a - \vec x \| \, \| \hat b - \vec y \| \nonumber \\
 \ \geq \ a_\perp^2 + b_\perp^2 + 2 \sqrt{1-\chi^2} a_\perp b_\perp \ \geq \ \chi^2, \quad
\ea
where the last inequality is due to Lemma~\ref{ref_lemma2}.

Note here that inequality~\eqref{ineq_lemma1_SI} can only be saturated if $\hat a, \hat b, \vec x$ and $\vec y$ are coplanar and if $\vec x$ is the orthogonal projection of $\hat a$ onto the direction $\hat x$ (so that the inequality in~\eqref{eq_proof_lemma1_SM} is saturated), and $\vec y$ is the orthogonal projection of $\hat b$ onto the direction $\hat y$.
\end{proof}

\medskip

For completeness, let us also introduce the following additional Lemma:

\setcounter{thm}{2}

\begin{lem}\label{ref_lemma3}

For any two \emph{orthogonal vectors} $\vec x$ and $\hat y$ of ${\cal E}$, such that $||\hat y|| = 1$, one has (with $b_\perp = \sqrt{1-(\hat b \cdot \hat y)^2}$):
\ba
\| \hat a - \vec x \|^2 + b_\perp^2 + 2 \sqrt{1 - \chi^2} \ \| \hat a - \vec x \| \, b_\perp \ \geq \ \chi^2. \label{ineq_lemma3_SI}
\ea

\end{lem}

\noindent The proof of this Lemma follows closely that of Lemma~\ref{ref_lemma1} above. Inequality~\eqref{ineq_lemma3_SI} can only be saturated if $\hat a, \hat b, \vec x$ and $\hat y$ are coplanar and if $\vec x$ is the orthogonal projection of $\hat a$ onto the direction $\hat x$.

\medskip

Lemmas~\ref{ref_lemma1} and~\ref{ref_lemma2} were used to prove our error-trade-off and error-disturbance relations~(\ref{ETR}--\ref{ETR_dimensionless}) and~(\ref{ETR_same_spectrum}), respectively (see the Methods section).
Lemma~\ref{ref_lemma3} can similarly be used to prove an error-disturbance relation when one does not assume that ${\cal A}$ must have the same spectrum as $A$. Namely, with the additional assumptions that $B^2 = {\cal B}^2 = \one$ and $\moy{B} = 0$, one has
\ba
\hspace{-2mm} \tilde\epsilon_{\cal A}^2 + \eta_{\cal B}^2 \Big( 1-\frac{\eta_{\cal B}^2}{4} \Big) + 2 \sqrt{1 - \tilde C_{\!AB}^2} \ \, \tilde\epsilon_{\cal A} \ \eta_{\cal B} \sqrt{1-\frac{\eta_{\cal B}^2}{4}} \ \geq \ \tilde C_{\!AB}^2 . \hspace{-8mm} \nonumber \\ \label{ETR_same_spectrum_B_only}
\ea
Figure~\ref{fig_comparison_3_ETRs} shows a comparison between the constraints imposed by each of our error-trade-off and error-disturbance relations~(\ref{ETR}--\ref{ETR_dimensionless}), \eqref{ETR_same_spectrum} and~(\ref{ETR_same_spectrum_B_only}).

\begin{figure}
\begin{center}
\epsfxsize=7.5cm
\epsfbox{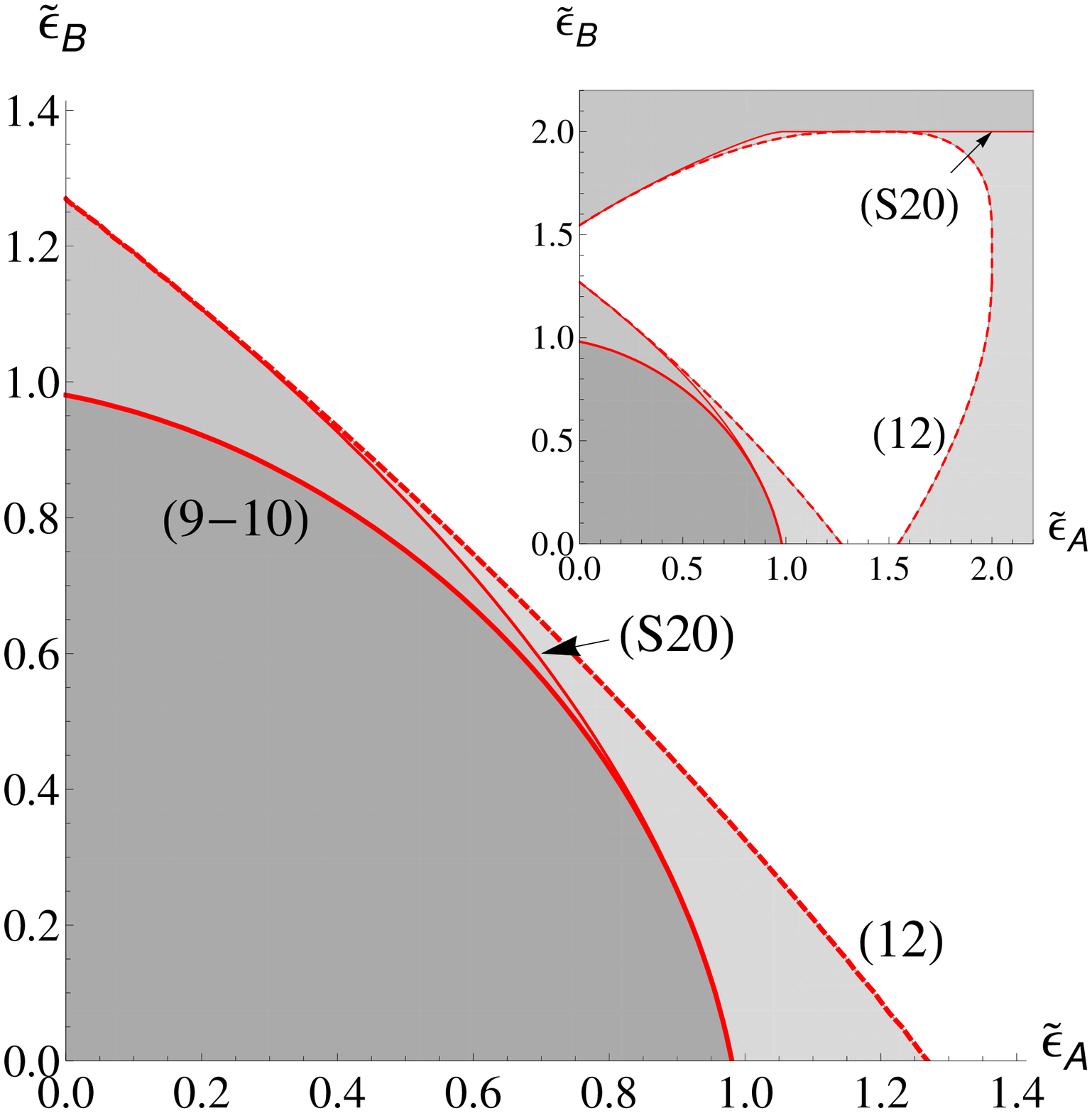}
\caption{{\bf Constraints imposed by our error-trade-off and error-disturbance relations~(\ref{ETR}--\ref{ETR_dimensionless}), (\ref{ETR_same_spectrum}) and~(\ref{ETR_same_spectrum_B_only})} (shown here for $\tilde C_{\!AB} = \sin \frac{7\pi}{16}$). Main figure: the shaded areas below the curves are inaccessible to the values of $(\tilde\epsilon_{\cal A}, \tilde\epsilon_{\cal B})$, depending on the assumptions under which the relations are derived (see text). Inset: in the case where $A^2 = B^2 = {\cal A}^2 = {\cal B}^2 = \one$ and $\moy{A} = \moy{B} = 0$, inequality~\eqref{ETR_same_spectrum} also bounds the values of $(\tilde\epsilon_{\cal A}, \tilde\epsilon_{\cal B}) = (\epsilon_{\cal A}, \eta_{\cal B})$ from above; similarly, inequality~\eqref{ETR_same_spectrum_B_only} also upper-bounds the values of $\tilde\epsilon_{\cal B} = \eta_{\cal B}$.}
\label{fig_comparison_3_ETRs}
\end{center}
\end{figure}

Let us finally mention that the necessary condition for the inequalities~\eqref{ineq_lemma2_SI}, \eqref{ineq_lemma1_SI} and~\eqref{ineq_lemma3_SI} in the three lemmas to be saturated may be helpful to inspire the choice of approximate measurements ${\cal A}$ and ${\cal B}$ if one wants to saturate our relations~(\ref{ETR}--\ref{ETR_dimensionless}), \eqref{ETR_same_spectrum} and~(\ref{ETR_same_spectrum_B_only}). ${\cal A}$ and ${\cal B}$ should indeed be chosen such that the real vectors $\vec x$ (or $\hat x$) and $\vec y$ (or $\hat y$) defined in Eq.~\eqref{def_vecx_vecy} of the Methods section are coplanar with $\hat a$ and $\hat b$, defined in Eq.~\eqref{def_veca_vecb} (and which are fixed by $A, B$ and $\ket{\psi}$). Furthermore, in order to saturate~(\ref{ETR}--\ref{ETR_dimensionless}) for instance, $\vec x$ and $\vec y$ should precisely be the orthogonal projections of $\hat a$ and $\hat b$; this is indeed ensured when one uses the optimal output functions $f_{opt}(m)$ and $g_{opt}(m)$ of Eqs.~(\ref{eqSI_f_opt}--\ref{eqSI_g_opt}).

\section{Tightness of our error-trade-off relations}
\label{secSI_proof_tightness}

We now show that our general error-trade-off relation~(\ref{ETR}--\ref{ETR_dimensionless}) is tight, and that so is our relation~\eqref{ETR_same_spectrum} when $A, B$ and $\ket{\psi}$ are---in addition to the assumptions $A^2 = B^2 = {\cal A}^2 = {\cal B}^2 = \one$ and $\moy{A} = \moy{B} = 0$---such that $|\moy{AB}| = 1$. 

In order to do so, we provide in each case explicit examples of approximate joint measurements ${\cal A}, {\cal B}$ (specified by their common eigenbases $\{\ket{m}\}$) saturating the bounds. These generalize the joint measurement strategies for qubits given in the main text. Note however that such optimal approximation strategies are not unique.

\subsection{Tightness of our error-trade-off relation~(\ref{ETR}--\ref{ETR_dimensionless})}

Interestingly, one does not need to use any ancillary system to saturate inequality~(\ref{ETR}--\ref{ETR_dimensionless}). As we show below, in the non-trivial case where $\Delta A \, \Delta B > 0$, this can indeed be done for instance by performing on $\ket{\psi}$ a projective measurement ${\cal M}$ with eigenvectors in ${\mathrm{Span}}\{\ket{\psi}, A\ket{\psi}, B\ket{\psi}\}$ (and where the remaining eigenvectors are orthogonal to $\ket{\psi}$, $A\ket{\psi}$ and $B\ket{\psi}$).

We define for convenience the observables $A_0 = [A-\moy{A}]/\Delta A$ and $B_0 = [B-\moy{B}]/\Delta B$ (such that $\moy{A_0} = \moy{B_0} = 0$ and $\moy{A_0^2} = \moy{B_0^2} = 1$). Note that from the Cauchy-Schwarz inequality applied to $A_0\ket{\psi}$ and $B_0\ket{\psi}$, one has $|\moy{A_0 B_0}| \leq 1$---which corresponds to the Schr\"odinger uncertainty relation~\cite{schrodinger_uncertainty_relation_SI}, and is a stronger version than the better-known Robertson uncertainty relation~\eqref{robertson_UR}.

\medskip

Let us consider first the case where $|\moy{A_0 B_0}| = 1$---in which the vectors $A_0\ket{\psi}$ and $B_0\ket{\psi}$ are linearly dependent---and define $\phi = \arg \, \moy{A_0B_0}$, $\phi_a = -\frac{\phi}{2}$ and $\phi_b = \frac{\phi}{2}$, such that $\moy{A_0B_0} = e^{i\phi}$ and $C_{\!AB} = \Delta A \, \Delta B \, \mathrm{Im} \moy{A_0B_0} = \Delta A \, \Delta B \sin \phi$. For an angle $\varphi$ and a real parameter $q \neq 0$, we define
\ba
\ket{\overline{m_1}} &=& (\one + q \, e^{i(\varphi - \phi_a)} A_0)\ket{\psi}, \label{def_m1_case1_saturate} \\
\ket{\overline{m_2}} &=& (\one - q^{-1} \, e^{-i(\phi_b - \varphi)} B_0)\ket{\psi}. \label{def_m2_case1_saturate}
\ea
Note that $\braket{\overline{m_1}}{\overline{m_2}} = 1 - e^{-i\phi} \moy{A_0B_0} = 0$. One can then define unit vectors $\ket{m_1}$ and $\ket{m_2}$ by normalizing the vectors $\ket{\overline{m_1}}$ and $\ket{\overline{m_2}}$, and complete the orthonormal basis $\{\ket{m_1}, \ket{m_2}\}$ of ${\mathrm{Span}}\{\ket{\psi}, A\ket{\psi}\} = {\mathrm{Span}}\{\ket{\psi}, B\ket{\psi}\}$ with some additional vectors $\ket{m}$ in ${\cal H}$, all orthogonal to $\ket{\psi}$, $A\ket{\psi}$ and $B\ket{\psi}$, so as to obtain a full orthogonal basis for ${\cal H}$; we then define the observable ${\cal M}$ to be a projective measurement onto this basis.

For a given outcome corresponding to an eigenvector $\ket{m}$, let us use the optimal approximations for the measurements of $A$ and $B$, i.e., let us output the values $f_{opt}(m)$ and $g_{opt}(m)$ given in~(\ref{eqSI_f_opt}--\ref{eqSI_g_opt}). According to Eqs.~(\ref{epsA2_optim}--\ref{epsB2_optim}), this leads to
\ba
\epsilon_{\cal A}^2 &=& \sum_{m = m_1, m_2} \, |\braket{m}{\psi}|^2 \left( \mathrm{Im} \, \frac{\sandwich{m}{A}{\psi}}{\braket{m}{\psi}} \right)^2 \nonumber \\
&=& \Delta A^2 \, \sum_{j = 1, 2} \, \frac{|\braket{\overline{m_j}}{\psi}|^2}{\braket{\overline{m_j}}{\overline{m_j}}} \left( \mathrm{Im} \, \frac{\sandwich{\overline{m_j}}{A_0}{\psi}}{\braket{\overline{m_j}}{\psi}} \right)^2 \nonumber \\
&=& \Delta A^2 \, \frac{1}{1+q^2} \left( \mathrm{Im} \, q \, e^{-i(\varphi - \phi_a)} \right)^2 \nonumber \\
&& + \, \Delta A^2 \, \frac{1}{1+q^{-2}} \left( \mathrm{Im} \, q^{-1} \, e^{-i(\varphi - \phi_a)} \right)^2 \nonumber \\[1mm]
&=& \Delta A^2 \, \sin^2 \!\left( \varphi - \phi_a \right) . \label{calculation_epsA}
\ea
$\epsilon_{\cal B}$ can be calculated in a similar way. One thus finds
\ba
\epsilon_{\cal A} = \Delta A \left| \sin \!\left( \varphi - \phi_a \right) \right|, \ \ 
\epsilon_{\cal B} = \Delta B \left| \sin \!\left( \phi_b - \varphi \right) \right|, \ \ 
\ea
which are independent of $q$.
By varying $\varphi$, chosen such that $\cos \phi \sin(\varphi{-}\phi_a) \sin(\phi_b{-}\varphi) \geq 0$, one obtains all values of $(\epsilon_{\cal A}, \epsilon_{\cal B})$ saturating the bound~(\ref{ETR}--\ref{ETR_dimensionless}).

\medskip

Consider now the case where $|\moy{A_0 B_0}| < 1$, and define $\chi = \frac{C_{\!AB}}{\Delta A \, \Delta B} = \mathrm{Im} \moy{A_0B_0}$, $\phi' = \arcsin \chi$, $\phi_a' = -\frac{\phi'}{2}$ and $\phi_b' = \frac{\phi'}{2}$. For an angle $\varphi \in [-\frac{|\phi'|}{2},\frac{|\phi'|}{2}]$ and two real parameters $q,r$, let us define the two complex coefficients
\ba
\alpha &=& q \, \cos (\varphi-\phi_a') + i \, r \, \sin (\varphi-\phi_a'), \\
\beta &=& r \, \cos (\phi_b'-\varphi) - i \, q \, \sin (\phi_b'-\varphi),
\ea
and the vector
\ba
\ket{\overline{m_1}} &=& (\one + \beta \, A_0 + \alpha \, B_0)\ket{\psi}.
\ea
Let us fix two more real parameters $s,t$ such that $s \, t = \braket{\overline{m_1}}{\overline{m_1}}/(1-|\moy{A_0 B_0}|^2)$, and define, with the operator $D = \sandwich{\overline{m_1}}{B_0}{\psi} A_0 - \sandwich{\overline{m_1}}{A_0}{\psi} B_0$,
\ba
\ket{\overline{m_2}} &=& (\braket{\overline{m_1}}{\overline{m_1}} \one + s \, D)\ket{\psi} - \ket{\overline{m_1}}, \\
\ket{\overline{m_3}} &=& (\braket{\overline{m_1}}{\overline{m_1}} \one - t \, D)\ket{\psi} - \ket{\overline{m_1}}.
\ea

One can then easily check, using in particular the facts that $\braket{\overline{m_1}}{\psi} = 1$, $\sandwich{\overline{m_1}}{D}{\psi} = \sandwich{\psi}{D}{\psi} = 0$ and $\sandwich{\psi}{D^\dagger D}{\psi} = (1 - |\moy{A_0 B_0}|^2) (\braket{\overline{m_1}}{\overline{m_1}} - 1)$, that the three vectors $\ket{\overline{m_1}}, \ket{\overline{m_2}}$ and $\ket{\overline{m_3}}$ are orthogonal. As before, one can normalize these, complete the thus obtained orthonormal basis $\{\ket{m_1}, \ket{m_2}, \ket{m_3}\}$ of ${\mathrm{Span}}\{\ket{\psi}, A\ket{\psi}, B\ket{\psi}\}$ with some additional vectors $\ket{m}$ in ${\cal H}$ (all orthogonal to $\ket{\psi}$, $A\ket{\psi}$ and $B\ket{\psi}$), and define the observable ${\cal M}$ to be a projective measurement onto this basis.

After some similar calculations as in~\eqref{calculation_epsA} (although a bit more tedious\footnote{One can use here in particular the facts that $\braket{\overline{m_2}}{\psi} = \braket{\overline{m_1}}{\overline{m_1}} - 1$, $\braket{\overline{m_2}}{\overline{m_2}} = (\braket{\overline{m_1}}{\overline{m_1}} - 1) [\braket{\overline{m_1}}{\overline{m_1}} + s^2 (1 - |\moy{A_0 B_0}|^2)]$, $\sandwich{\overline{m_2}}{A_0}{\psi} = s \alpha (1 - |\moy{A_0 B_0}|^2) - \sandwich{\overline{m_1}}{A_0}{\psi}$, $\sandwich{\overline{m_2}}{B_0}{\psi} = -s \beta (1 - |\moy{A_0 B_0}|^2) - \sandwich{\overline{m_1}}{B_0}{\psi}$, and similar relations with $\ket{\overline{m_3}}$ (and with $s$ replaced by $-t$); furthermore, $\braket{\overline{m_1}}{\overline{m_1}} - 1 = (1-\chi^2) (q^2 + r^2 + 2 q r \frac{\mathrm{Re} \moy{A_0 B_0}}{\sqrt{1-\chi^2}})$, $(\mathrm{Im} \sandwich{\overline{m_1}}{A_0}{\psi})^2 = (1-\chi^2) (q^2 + r^2 + 2 q r \frac{\mathrm{Re} \moy{A_0 B_0}}{\sqrt{1-\chi^2}}) \sin^2(\varphi{-}\phi_a') - (1 - |\moy{A_0 B_0}|^2) (\mathrm{Im} \alpha)^2$ and $(\mathrm{Im} \sandwich{\overline{m_1}}{B_0}{\psi})^2 = (1-\chi^2) (q^2 + r^2 + 2 q r \frac{\mathrm{Re} \moy{A_0 B_0}}{\sqrt{1-\chi^2}}) \sin^2(\phi_b'{-}\varphi) - (1 - |\moy{A_0 B_0}|^2) (\mathrm{Im} \beta)^2$.}), one finds again, from Eqs.~(\ref{epsA2_optim}--\ref{epsB2_optim}) (i.e. for the optimal approximations of the measurements of $A$ and $B$),
\ba
\epsilon_{\cal A} = \Delta A \left| \sin \!\left( \varphi - \phi_a' \right) \right|, \ \ 
\epsilon_{\cal B} = \Delta B \left| \sin \!\left( \phi_b' - \varphi \right) \right|, \quad \ 
\ea
independently of $q,r,s$ and $t$.
By varying $\varphi \in [-\frac{|\phi'|}{2},\frac{|\phi'|}{2}]$, one obtains again all values of $(\epsilon_{\cal A}, \epsilon_{\cal B})$ saturating the bound~(\ref{ETR}--\ref{ETR_dimensionless}).

\medskip

Let us finally mention the case where $\Delta A \, \Delta B = 0$ (which implies $C_{\!AB} = 0$). In such a case, one can have $\epsilon_{\cal A} = \epsilon_{\cal B} = 0$, saturating again the (then trivial) inequality~(\ref{ETR}): this is indeed obtained for instance by defining ${\cal A} = \moy{A}$ and ${\cal B} = B$ if $\Delta A = 0$, or ${\cal A} = A$ and ${\cal B} = \moy{B}$ if $\Delta B = 0$.

\subsection{Tightness of our error-disturbance relation~\eqref{ETR_same_spectrum} \newline (valid for $A^2 = B^2 = {\cal A}^2 = {\cal B}^2 = \one$ and $\moy{A} = \moy{B} = 0$), \newline when $|\moy{AB}| = 1$}
\label{secSI_proof_tightness2_ETR_same_spectrum}

Let us now turn to the case where ${\cal A}$ and ${\cal B}$ are assumed to have the same spectrum as $A$ and $B$. In such a case one can in general no longer choose the output functions $f(m)$ and $g(m)$ for the approximations of $A$ and $B$ to be the optimal ones, prescribed by Eqs.~(\ref{eqSI_f_opt}--\ref{eqSI_g_opt}). If $A$ and $B$ are dichotomic observables with eigenvalues $\pm 1$, then $f(m)$ and $g(m)$ are also bound to take values $\pm 1$; from Eq.~\eqref{eqSI_epsA2_to_optimize}, one can see that the optimal choice to minimize $\epsilon_{\cal A}$ and $\epsilon_{\cal B}$ is now to choose
\ba
f(m) = \mathrm{sign} \! \left( \! \mathrm{Re} \frac{\sandwich{m}{A}{\psi}}{\braket{m}{\psi}} \! \right) \! , \ \ 
g(m) = \mathrm{sign} \! \left( \! \mathrm{Re} \frac{\sandwich{m}{B}{\psi}}{\braket{m}{\psi}} \! \right) \! . \hspace{-5mm} \nonumber \\ \label{f_g_pm_1}
\ea
Using~\eqref{eqSI_epsA2_to_optimize}, this leads to
\ba
\epsilon_{\cal A}^2 &=& \! \moy{A^2} \! + \! \sum p(m) f(m)^2 \! - \! 2 \! \sum p(m) f(m) \, \mathrm{Re} \frac{\sandwich{m}{A}{\psi}}{\braket{m}{\psi}} \hspace{-2mm} \nonumber \\
&=& 2 - 2 \sum p(m) \Big| \mathrm{Re} \, \frac{\sandwich{m}{A}{\psi}}{\braket{m}{\psi}} \Big| \, , \label{epsA2_optim_pm1}
\ea
and similarly,
\ba
\epsilon_{\cal B}^2 &=& 2 - 2 \sum p(m) \Big| \mathrm{Re} \, \frac{\sandwich{m}{B}{\psi}}{\braket{m}{\psi}} \Big| \, . \label{epsB2_optim_pm1}
\ea

When $|\moy{AB}| = 1$, let us use the eigenvectors defined from~(\ref{def_m1_case1_saturate}--\ref{def_m2_case1_saturate}), with $q=\pm 1$. We now find, from~(\ref{epsA2_optim_pm1}--\ref{epsB2_optim_pm1}),
\ba
\epsilon_{\cal A} = \sqrt{ 2{-}2\big|\!\cos (\varphi{-}\phi_a) \big| }, \ \ 
\epsilon_{\cal B} = \sqrt{ 2{-}2\big|\!\cos (\phi_b{-}\varphi) \big| } \quad \ \
\ea
(note that if $\cos (\varphi{-}\phi_a) \geq 0$ and $\cos (\phi_b{-}\varphi) \geq 0$, one gets $\epsilon_{\cal A} = 2|\sin\frac{\varphi{-}\phi_a}{2}|$ and $\epsilon_{\cal B} = 2|\sin\frac{\phi_b{-}\varphi}{2}|$).
By varying $\varphi$, such that $\cos \phi \sin(\varphi{-}\phi_a) \sin(\phi_b{-}\varphi) \geq 0$, one obtains all minimal values of $(\epsilon_{\cal A}, \epsilon_{\cal B})$ saturating the bound~\eqref{ETR_same_spectrum}.

Furthermore, as can be seen for instance from Figure~\ref{fig_comparison_3_ETRs} (inset), inequality~\eqref{ETR_same_spectrum} also sets upper bounds on the possible values of $(\epsilon_{\cal A}, \epsilon_{\cal B})$. By using the same eigenvectors as above but changing the signs of $f(m)$ and $g(m)$ in Eq.~\eqref{f_g_pm_1}, and possibly mixing such strategies, $(\epsilon_{\cal A}, \epsilon_{\cal B})$ can attain all possible values along the contour of the region restricted by inequality~\eqref{ETR_same_spectrum}.
This shows that our error-disturbance relation~\eqref{ETR_same_spectrum} (valid for $A^2 = B^2 = {\cal A}^2 = {\cal B}^2 = \one$ and $\moy{A} = \moy{B} = 0$) is tight when\footnote{Note that under the conditions $A^2 = B^2 = \one$ and $\moy{A} = \moy{B} = 0$, $|\moy{AB}|=1$ always holds for qubits (as in the case considered in the example of the main text): indeed, both $A\ket{\psi}$ and $B\ket{\psi}$ are orthogonal to $\ket{\psi}$; as ${\cal H}$ is of dimension 2, they are therefore linearly dependent, and $|\moy{AB}|^2 = \moy{A^2} \moy{B^2} = 1$.} $|\moy{AB}| = 1$; its tightness in the case $|\moy{AB}| < 1$ is left as an open problem.

\medskip

One may also wonder if inequality~\eqref{ETR_same_spectrum} remains tight (when $|\moy{AB}| = 1$) in the specific error-disturbance scenario, where one imposes that ${\cal A}$ and ${\cal B}$ have the particular forms ${\cal A} = U^\dagger (\one \otimes M_A) U$ and ${\cal B} = U^\dagger (B \otimes \one) U$ (see main text). The answer is positive: one can indeed transform the previous measurement strategy so that $A$ and $B$ are estimated from measurements on two separate systems. Intuitively, one needs to copy some information on the quantum state $\ket{\psi}$ onto the ancillary system, and rotate the first system so that the measurement in the eigenbasis $\{\ket{m_1}, \ket{m_2}, \ldots\}$ considered before becomes a measurement of $B$ directly. Formally, one can define $U = (U_{\textrm{R}}\otimes\one).U_{\textrm{copy}}$ and $M_A = \sum f(m) \ket{m}\!\bra{m}$, where $U_{\textrm{copy}}$ is a unitary such that $U_{\textrm{copy}} \ket{m,\xi} = \ket{m,m}$ for all basis vectors $\ket{m}$ and $U_{\textrm{R}}$ is a unitary such that $U_{\textrm{R}}^\dagger B U_{\textrm{R}} = \sum g(m) \ket{m}\!\bra{m}$ (where the non-yet-prescribed values $g(m)$ are chosen so that the numbers of $+1$ and $-1$ values are the same as the numbers of $+1$ and $-1$ eigenvalues of $B$, for such a unitary $U_{\textrm{R}}$ to exist). One then obtains the same values for $\epsilon_{\cal A}$ and $\epsilon_{\cal B}$ as before, with the direct measurement of $\ket{\psi}$ in the eigenbasis $\{\ket{m}\}$.

\medskip

Let us finally mention that one can prove in very similar ways that our error-disturbance relation~(\ref{ETR_same_spectrum_B_only}), valid under the assumptions that $B^2 = {\cal B}^2 = \one$ and $\moy{B} = 0$, is also tight when $|\moy{A_0B}| = 1$.

\section{Ozawa's ``uncertainty relation'' follows from our error-trade-off relation}
\label{SecSI_ozawa_from_new_relation}

We show in this last part that Ozawa's ``uncertainty relation'' for joint measurements, Eq.~\eqref{ozawa_relation}~\cite{ozawa04_SI}, follows from our error-trade-off relation~\eqref{ETR}.
Let us indeed write:
\ba
&& \hspace{-5mm} (\epsilon_{\cal A} \, \epsilon_{\cal B} + \Delta B \, \epsilon_{\cal A} + \Delta A \, \epsilon_{\cal B})^2 \nonumber \\[1mm]
&& \hspace{-3mm} \geq (\Delta B \epsilon_{\cal A}{+}\Delta A \epsilon_{\cal B})^2 = \Delta B^2 \epsilon_{\cal A}^2{+}\Delta A^2 \epsilon_{\cal B}^2{+}2 \Delta A \Delta B \, \epsilon_{\cal A} \epsilon_{\cal B} \nonumber \\
&& \hspace{-1mm} \geq \Delta B^2 \epsilon_{\cal A}^2{+}\Delta A^2 \epsilon_{\cal B}^2{+}2 \sqrt{\Delta A^2 \Delta B^2{-}C_{\!AB}^2} \, \epsilon_{\cal A} \epsilon_{\cal B} \, \geq \, C_{\!AB}^2 , \nonumber \\
\label{ozawa_follows}
\ea
where the last inequality is precisely our error-trade-off relation~\eqref{ETR}. After taking the square-root of the above expressions, we obtain Ozawa's relation~\eqref{ozawa_relation}.

Hence, one can see that Ozawa's relation is sub-optimal in two ways. Firstly, one can simply drop the first product term, $\epsilon_{\cal A} \epsilon_{\cal B}$---which, interestingly, is precisely the term in the Heisenberg-Arthurs-Kelly relation~\eqref{heisenberg_relation}. Secondly, one can decrease the factor $\Delta A \, \Delta B$ on the second line of Eq.~\eqref{ozawa_follows} down to $\sqrt{\Delta A^2 \Delta B^2 - C_{\!AB}^2}$. We note also that Ozawa's relation can only be saturated for $\epsilon_{\cal A} = 0$ or $\epsilon_{\cal B} = 0$, as otherwise the first inequality in~\eqref{ozawa_follows} is strict.

Similarly, Ozawa's error-disturbance ``uncertainty relation''~\eqref{ozawa_error_disturbance_relation}~\cite{ozawa03_SI} also follows from our relation~\eqref{ETR}. Recall that in the case where $A^2 = B^2 = {\cal A}^2 = {\cal B}^2 = \one$ and $\moy{A} = \moy{B} = 0$, the strictly stronger relation~\eqref{ETR_same_spectrum} holds; in that case Ozawa's relation~\eqref{ozawa_error_disturbance_relation} cannot be saturated, except in the trivial situation where $C_{\!AB} = 0$ (which allows for $\epsilon_{\cal A} = \eta_{\cal B} = 0$).

\medskip

We note finally that Hall derived in Ref.~\cite{hall04_SI} a very similar ``joint-measurement uncertainty relation'' to Ozawa's relation~\eqref{ozawa_relation}, where the standard deviations $\Delta A, \Delta B$ of $A$ and $B$ in Ozawa's relation are replaced by the standard deviations $\Delta {\cal A}, \Delta {\cal B}$ of ${\cal A}$ and ${\cal B}$. Despite the claim in~\cite{hall04_SI}, Hall and Ozawa's relations are however not equivalent. Hall's relation indeed involves more quantities that depend on the particular choice of ${\cal A}$ and ${\cal B}$ (namely, $\Delta {\cal A}$ and $\Delta {\cal B}$), and not only on $A, B$ and $\ket{\psi}$. As a consequence, Hall's inequality does not simply follow from our error-trade-off relation~\eqref{ETR}---it can even give a stronger constraint on $\epsilon_{\cal A}$ and $\epsilon_{\cal B}$, for fixed values of $\Delta {\cal A}$ and $\Delta {\cal B}$. As it is the case with Ozawa's inequality, we nevertheless also expect a strictly stronger relation than Hall's to hold, when one considers all quantities $\Delta {\cal A}, \Delta {\cal B}, \epsilon_{\cal A}$ and $\epsilon_{\cal B}$ (possibly in addition to $\Delta A$ and $\Delta B$, which are fixed by $A, B$ and $\ket{\psi}$).

\medskip

\noindent \emph{Note added.}
Hall's relation has very recently also been investigated experimentally~\cite{weston12}, together with Ozawa's relation~\eqref{ozawa_relation} and another relation derived by the authors of~\cite{weston12}. Their experiment also demonstrated a violation of the Heisenberg-Arthurs-Kelly relation~\eqref{heisenberg_relation}.


\end{document}